\documentclass[journal,draftclsnofoot,onecolumn,12pt]{IEEEtran}
\usepackage{cite}
\usepackage{amsmath,amssymb,amsfonts,amsthm}
\usepackage{mathrsfs}
\usepackage{array}
\usepackage{tabularx}
\usepackage{multirow}
\usepackage{adjustbox}
\usepackage{lipsum}
\usepackage{booktabs}
\usepackage{color,colortbl}
\definecolor{lavender}{rgb}{0.9, 0.9, 0.98}
\usepackage{mathtools}
\usepackage{subfigure}
\usepackage{mathtools}
\usepackage{graphicx}
\usepackage{bbm}
\usepackage[utf8]{inputenc}
\usepackage{soul}
\usepackage{algpseudocode}
\usepackage[linesnumbered,ruled,vlined]{algorithm2e}
\SetKwInput{KwInput}{Input}                
\SetKwInput{KwOutput}{Output}
\SetKwInput{KwInit}{Initialization}

\newtheorem{lemma}{Lemma}
\usepackage{textcomp}
\usepackage[normalem]{ulem}

\begin{document}

\title{RIS-NOMA integrated low-complexity transceiver architecture: Sum rate and energy efficiency perspective}

\author{Kali Krishna Kota,~\IEEEmembership{Student Member,~IEEE,} Praful D. Mankar,~\IEEEmembership{Member,~IEEE,}
}

\markboth{Journal of \LaTeX\ Class Files,~Vol.~14, No.~8, January~2024}%
{Shell \MakeLowercase{\textit{et al.}}: A Sample Article Using IEEEtran.cls for IEEE Journals}

\IEEEpubid{0000--0000/00\$00.00~\copyright~2021 IEEE}

\maketitle
\begingroup
\allowdisplaybreaks
\begin{abstract}
This paper aims to explore RIS integration in a millimeter wave (mmWave) communication system with low-complexity transceiver architecture under imperfect channel state information (CSI) assumption. Motivated by this, we propose a RIS-aided system with a fully analog architecture at the base station (BS). However, to overcome the drawback of single-user transmission due to the single RF chain in the analog architecture, we propose to employ NOMA to enable multi-user transmission. For such a system, we formulate two problems to obtain the joint transmit beamformer, RIS phase shift matrix, and power allocation solutions that maximize  sum rate and energy efficiency such that the minimum rate for each user is satisfied. However, both problems are intractable due to 1) the fractional objective, 2) non-convex minimum rate and unit modulus RIS phase shift constraints, and 3) the coupled optimization variables. Hence, we first tackle the fractional objectives of both problems by reformulating the sum rate and energy efficiency maximization problems into equivalent quadratic forms using the quadratic transform. On the other hand, we employ successive convex approximation and the semi-definite relaxation technique to handle the non-convex minimum rate and unit modulus constraint of the RIS phase shifts, respectively. However, the problems remain non-convex due to the coupled optimization variables. Thus, we propose an alternating optimization-based algorithm that iterates over the transmit beamformer, power allocation, and RIS phase shift subproblems. Further, we also show that the quadratic reformulation is equivalent to the weighted mean square error-based reformulation for the case  of  sum rate maximization problem. Our numerical results show that the proposed RIS-NOMA integrated analog architecture system outperforms the optimally configured fully digital architecture in terms of sum rate at low SNR and energy efficiency for a wide range of SNR while still maintaining low hardware complexity and cost. Finally, we present the numerical performance analysis  of the RIS-NOMA integrated low-complexity system for various system configuration parameters. 
\end{abstract}
\begin{IEEEkeywords}
Reconfigurable Intelligent Surfaces, Analog beamforming, Sum rate, Energy Efficiency, NOMA.
\end{IEEEkeywords}

\section{Introduction}
\label{sec:introduction}
The design of next-generation cellular systems is driven by the need to provide reliable connectivity with high data rates and low latency. To achieve these goals, millimeter wave (mmWave) band (30-300 GHz)-based communication, commonly referred to as mmWave communication \cite{Intro_mmWave_1}, has been considered in the deployment of 5G networks. The abundant spectrum available in this band has the potential to enable high data rates and minimize interference in wireless communication \cite{Intro_mmWave_2}, \cite{Intro_mmWave_3}. However, mmWave signals encounter several propagation challenges due to significant atmospheric absorption and penetration losses. These losses result in unreliable link performance, especially in the absence of line-of-sight (LoS) \cite{Intro_mmWave_2}-\cite{Intro_mmWave_4}. This restricts the mmWave band for short-range communication. Interestingly, the emergence of reconfigurable intelligent surfaces (RIS) has presented an elegant solution to overcome the limitations imposed by such severe path loss. RIS is a planar array with multiple reflecting elements distributed uniformly across its surface \cite{Intro_RIS_1,Intro_RIS_3}. These elements can collectively reorient the incident electromagnetic waves in the desired direction, effectively manipulating the propagation environment. Thus, by leveraging RIS technology, a virtual line-of-sight (LoS) link can be established between the transmitter and receiver such that it can aid communication in the mmWave band. Motivated by this, our paper considers the optimal design of RIS-assisted mmWave communication. Further, we assume passive RIS in order to reduce power consumption and design complexity. 

On the other hand, an essential aspect in the design of multi-antenna communication systems is to carefully choose the interface between the high-dimensional antenna array and  baseband processing unit (BBU). A variety of architectures proposed in the literature for this interface \cite{DB_1,DB_3} can be broadly categorized into 1) fully digital, 2) fully analog, and 3) hybrid. In the fully digital architecture, each antenna at the base station has a dedicated radio frequency (RF) chain \cite{62}, which conveniently allows multi-user precoding in the digital domain to harness maximum spatial multiplexing gain and thus can provide benchmark performance. It can transmit information to as many users as there are the number of antennas \cite{DB_4,DB_5,DB_6,DB_7}. However, this architecture has high implementation cost and power consumption as each antenna needs a dedicated RF chain, which includes mixed signal components like low noise amplifiers, digital to analog converters, and mixers \cite{DB_2}. A promising alternative is the hybrid architecture \cite{13} that divides the transmit processing into digital baseband and analog RF processing resulting in decreased number of RF-chains, thereby reducing the hardware complexity \cite{14}, while still providing reasonable multiplexing gains. Nevertheless, this architecture has the disadvantage of high computational complexity arising due to the feedback overhead between the analog and digital domains \cite{Convex_QP_SDP}. To further reduce the hardware complexity, a fully analog architecture is another solution where a single RF chain connects all the antennas to the BBU \cite{FA_1,FA_2}. {\em However, analog architecture lacks the ability to perform digital domain precoding, which in turn imposes the restriction of single-user transmission.} Additionally, it also has the disadvantage of the unit-modulus constraint associated with the analog beamformer \cite{DB_2}. The mathematical intractability arising due to this unit modulus constraint can be overcome by adopting a double-phase shifter (DPS)-based analog beamformer \cite{DPS_Ref1,DPS_Ref2}. The use of DPS not only simplifies the mathematical formulation but also increases link performance in terms of the antenna array gain. Next, in order to overcome the disadvantage of single-user transmission, we can employ \textit{non-orthogonal multiple access} (NOMA) to facilitate multi-user transmission and obtain multiplexing gain. Specifically, employing power-domain NOMA can exploit different user channel conditions and transmit information corresponding to multiple users at different power levels using superposition coding so that each user can decode its corresponding message using successive interference cancellation (SIC) \cite{19,20}. 
Since the multi-user signal superpositioning is done in the BBU, NOMA can facilitate the multi-user transmission using a single RF chain-based fully analog architecture \cite{61,40}. As a result, the usage of NOMA will improve spectral efficiency \cite{21,22} by unveiling the spatial multiplexing gain in the fully analog architecture. Besides, the advantage of NOMA over orthogonal multiple access (OMA) is well-proven in the literature in terms of user fairness, spectral efficiency, and energy efficiency \cite{23,24}. Nonetheless, the performance of NOMA depends on the disparity in the channel conditions observed by users. This is where RIS integration can be useful, as it can partially control the propagation environment of wireless channels. Such flexibility of controlling channel can aid in optimizing the performance of NOMA-based system \cite{RISmmWaveNOMA_RobSchober},\cite{RIS_NOMA_Interplay_RobSchober,NOMA_NGMA_6G_RobSchober}. \textit{Further, it is to be noted that careful integration of RIS can dramatically compensate for the performance loss of the fully analog architecture while still maintaining low hardware complexity and cost/power consumption. Motivated by this, we focus on the optimal design of the RIS-assisted NOMA-mmWave system with a fully analog architecture in this paper.}

\subsection{Related works}  
The design and performance of RIS-aided NOMA communication systems has been extensively studied in the literature. For example, the authors of \cite{RIS_NOMA1_MarcoDR} and \cite{RIS_NOMA2_DC} focus on the optimal beamforming problem of the assisted NOMA-MISO downlink system. The former achieves the optimal design by maximizing the minimum decoding SINR to ensure rate performance and user fairness, whereas the latter aims to minimize the total transmit power. Both works decouple the user ordering and the beamforming problems and then solve the transmit beamformer and phase shift matrix subproblems iteratively. Further, in order to obtain the solution to this joint beamformer problem, a unified algorithm based on block coordinate descent (BCD) has been proposed for \cite{RIS_NOMA1_MarcoDR}, and difference-of-convex (DC)-based method for \cite{RIS_NOMA2_DC}. Similarly, the authors of \cite{RIS_NOMA3_HyNOMA} propose the design of a RIS-aided MISO system with NOMA that pairs well-separated users. Next, the joint transmit beamformer and RIS phase shift problem are solved using an alternating optimization-based technique. 
The above-mentioned works numerically show that integrating RIS into NOMA-based systems improves the system performance significantly. 
Similarly, the authors of \cite{RIS_mmWave2_HyB} consider integrating RIS into the mmWave downlink system with hybrid beamforming architecture. To jointly optimize the hybrid beamformer and RIS phase shift matrix of such a system, the authors formulate a power minimization problem. More specifically, a penalty-based algorithm is proposed where the beamforming and RIS phase shift matrix subproblems are solved iteratively using BCD. Further, they propose three different techniques to handle the unit modulus constraints, namely, 1) alternating optimization, 2) Riemannian conjugate gradient algorithm, and 3) successive convex approximation-based algorithm. The authors of \cite{RIS_mmWave2_SR} minimize the sum MSE of the symbols decoded of a RIS-aided MU-MIMO system with a hybrid beamforming architecture. Similar to \cite{RIS_mmWave2_HyB}, the authors propose an MM-based algorithm where they tackle the unit modulus constraint of the beamformer and the RIS phase shift matrix using the accelerated Riemannian gradient algorithm. Additionally, they also propose an enhanced regularized zero-forcing scheme that enables performance analysis in both low and high SNR regimes.   
Next, the authors in \cite{EE_mmWave_RIS_HyB} considered an energy-efficient RIS-aided mmWave multi-user MIMO downlink system by considering a hybrid beamforming architecture. Specifically, they adopted a sequential approach to obtain the transmit power and the number of active RF chains for a given transmit precoder, combiner, and RIS phase shift matrix.

Further, \cite{RIS_NOMA_mmWave_HyB_MarcoDR} analyzes a RIS-aided mmWave-NOMA downlink system with a hybrid architecture. The authors focus on maximizing the achievable  sum rate   while ensuring minimum rate constraints for each user. To solve this problem, the authors propose an alternating optimization algorithm that solves the RIS phase shift, hybrid beamforming, and power allocation problem subproblems iteratively. Specifically, they solve the RIS phase shift subproblem and the hybrid beamforming subproblem using alternating manifold optimization and successive convex approximation (SCA)-based methods. Similarly, in \cite{RIS_NOMA_mmWave_Letter}, the authors solve the joint transmit beamformer, phase shift matrix, and power allocation problem based on alternating optimization and SCA that maximizes the  sum rate. Furthermore, the authors of \cite{RIS_NOMA_mmWave_UL} focus on maximizing the energy efficiency of a RIS-aided mmWave NOMA uplink network with an analog beamforming architecture. The authors utilize BCD to solve the power allocation and beamforming problems jointly. In addition, \cite{RIS_NOMA_mmWave_Lens} presents the joint beamformer, phase shift matrix, and power allocation problem that maximizes the weighted  sum rate   of a massive MIMO system with multiple single-antenna users. The proposed system model adopts a discrete lens array that is often referred to as the continuous aperture phased (CAP)-MIMO architecture. 
Finally, the authors of \cite{EE_vs_SE} presented the energy efficiency vs. spectral efficiency tradeoff in a RIS-assisted mmWave NOMA downlink system. To achieve this, the authors formulated a multi-objective optimization problem that maximizes the sum rate and minimizes the power consumption in order to obtain the joint power allocation, active and passive beamforming solution.


In summary, the aforementioned works show that integration of RIS into (with/without NOMA) communication systems provides performance enhancement, however a majority of these works mostly consider fully digital or hybrid beamforming architecture. These works focus on the optimal design of joint beamforming, RIS phase shifts, and power allocation such that the performance of key metrics such as  sum rate, energy efficiency, or user fairness is maximized. Further, the above works propose algorithms and provide numerical insights into the system's performance, verifying that RIS enhances the performance of communication systems. \textit{Unlike the above works, in this paper, we take a step beyond and investigate if this huge RIS gain can relax the architectural complexity of the transmitter.} Thus, it is meaningful to explore the benefits of integrating RIS into less complex systems such as fully analog architecture. In this context, to the best of our knowledge, there is no work studying the performance of RIS-aided mmWave NOMA networks with a fully analog beamforming architecture, which is set to be the main goal of this paper.

\subsection{Contributions}
In this paper, we investigate the performance of a low-complexity transceiver system with a fully analog architecture that is integrated with RIS and is capable of performing NOMA transmission. In particular, we harness the capabilities of  1) NOMA to enable multi-user transmission (and thus improve spatial multiplexing capabilities of a single RF-chain system) and 2) RIS to enhance signal reception quality and improve performance. We study the sum rate  and energy efficiency performances of such a RIS-NOMA-aided fully analog architecture 
under imperfect knowledge of CSI. Towards this, we formulate sum rate  and energy efficiency maximization problems to obtain the joint power allocation, transmit beamformer, and RIS phase shift matrix solution while ensuring minimum rates to all the users. The main contributions of this paper are summarized below
\begin{itemize}
    \item The sum rate  and energy efficiency maximization problems are non-convex due to 1) the fractional nature of the SINR term, 2) the non-convex minimum rate and unit modulus RIS phase shift constraints, and 3) the coupled optimization variables. Thus, to tackle these problems, we propose a three-step framework: 1) simplify the objective function by using the quadratic transform to decouple the numerator and denominator in SINR terms, 2) employ SCA and SDR to relax the non-convex minimum rate constraint and the unit modulus phase shifts constraints, respectively, and 3) adopt alternating optimization (AO)-based iterative algorithm to decouple the transmit beamformer, RIS phase shift matrix and power allocation variables.  
    \item We analytically establish the equivalence between the quadratic transform and weighted MSE minimization under the transmit beamformer and the RIS phase shifts sub-problems. 
    \item An extensive numerical analysis reveals the following key system design insights: 
    1) The proposed system outperforms the optimally configured fully digital architecture in terms of sum rate at low SNR and in terms of energy efficiency for a wide range of SNR while still maintaining low hardware complexity and cost,
    2) the sum rate performance of the proposed system is limited by noise and CSI error variance in the low and high SNR regions, respectively. These limits can be pushed by increasing the number of RIS elements, 
    3) the energy efficiency performance is almost independent to the variation of CSI error variance at low SNR,
    and 
    4) RIS NOMA-based fully analog architecture provides significantly higher sum rate  compared to the RIS OMA-based fully analog architecture.
\end{itemize}

{\em Notations:} $a^*$, $\Re\{a\}$, and $|a|$ represent the conjugate, real part, and absolute value of $a$. $\left\lVert \mathbf{a} \right\lVert$ and $\mathbf{a}_i$ are the norm and the $i$-th element of vector $\mathbf{a}$, whereas $\mathbf{A}^T$, $\mathbf{A}^H$, $\|\mathbf{A}\|_F$, ${\rm{trace}}(\mathbf{A})$,  $\mathbf{A}_{i,:}$, $\mathbf{A}_{:,i}$ and $\mathbf{A}_{ij}$ are the transpose, Hermitian, Frobenius norm, trace, $i$-th row, $i$-th column and $ij$-th element of the matrix $\mathbf{A}$, respectively. The notation $\mathbb{C}^{M\times N}$ is the set of  $M \times N$ complex matrices, ${\rm I_M}$ is $M\times M$ identity matrix and $\mathbf{1}_{\rm M}$ is a $M\times 1$ vector with unit elements. $\mathbf{v_A}$ and $\lambda_\mathbf{A}$ are the principal eigenvector and eigenvalue of  $\mathbf{A}$.  $\odot$ is the hadamard product, $\rm{diag}(\mathbf{a})$ is a diagonal matrix such that vector $\mathbf{a}$ forms its diagonal, and $\mathcal{CN}(\boldsymbol{\mu},\mathbf{K})$ denotes complex  Gaussian distribution with mean $\boldsymbol{\mu}$ and covariance matrix $\mathbf{K}$.

\section{System model}
\begin{figure}[h!]
    \centering  \includegraphics[width=.5\columnwidth]{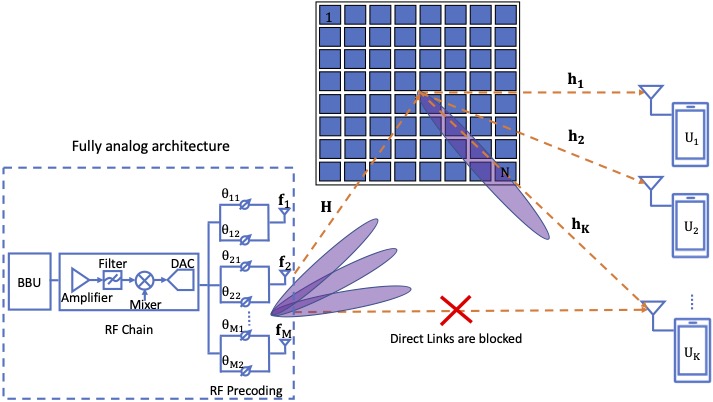}
    \caption{Illustration of the system model}
    \label{fig:enter-label}
\end{figure}
We consider a RIS-aided mmWave-NOMA downlink system with a fully analog architecture wherein BS having $M$ antennas is connected to the BBU via a single RF-chain. 
The BS transmits a superposition coded symbol $x$ to $K$ single antenna users such that 
\begin{equation}\label{SupCoded}
   x = \sum_{k=1}^{K}\sqrt{p_k} s_k, 
\end{equation}
where $ s_k$ is the symbol corresponding to the $k\text{-th}$ user, and $p_k$ is the allocated power to the $k\text{-th}$ user. 
We assume that there is no direct link between the BS and users, and the users receive the signal via a planar RIS that has $N_H$ and $N_V$ elements along the horizontal and vertical directions, respectively, such that $N_H N_V = N$. Thus, the signal received by the $k\text{-th}$ user can be written as 
\begin{align}
    \Tilde{y}_k= \mathbf{h}_{k}^H \mathbf{\Phi} \mathbf{H}\mathbf{f}x+n_k,\label{Rx_Signal} 
\end{align} 
where $\mathbf{H}\in {\mathbb C}^{N \times M}$ is the channel between the BS and RIS, $\mathbf{h}_{k}\in{\mathbb C}^{N \times 1}$ is the channel between the $k\text{-th}$ user and RIS, $n_k  \sim {\mathcal{CN}}(0,\sigma^2)$ is the noise at the $k\text{-th}$ user, and  $\mathbf{\Phi} = \rm{diag(\boldsymbol{\psi})}$ represents the RIS phase shift matrix. Due to the passive RIS assumption, the elements of phase shift vector $\boldsymbol{\psi}$ must satisfy the unit modulus constraint i.e., $|\boldsymbol{\psi}_i| = 1; \forall i = 1, \cdots N$. Next, $\mathbf{f} \in {\mathbb C}^{M\times 1}$ is the analog transmit beamforming vector connected to double phase shifters. Thus, each elements of $\mathbf{f}$ should satisfy 
\begin{align}
    |{\bf f}_{i}| \leq \max_{\theta_{i1}, \theta_{i2} \in [-\pi, \pi]} \quad \left|e^{j\theta_{i1} } + e^{j \theta_{i2}}\right| =2,
    \end{align}
   for $i=1,\dots,M$ where $\theta_{ij}$ is the phase of the $j$-th phase shifter associated with the $i$-th antenna. Further, we also assume the unit norm beamformer i.e., $\|{\bf f}\|^2 = 1$ in order to ensure the total transmit power $P_s$ constraint.

At the receiver end, we assume that each user decodes its corresponding symbol perfectly by employing successive interference cancellation (SIC). Moreover, we consider that the user decoding is performed in the ascending order of the channel gains.
The signal at the $k\text{-th}$ user after perfect SIC is given as
 \begin{align} 
y_k &= \Tilde{y}_k - \sum\nolimits_{i=1}^{k-1}\Tilde{y}_i \nonumber \\
  &= \mathbf{h}_{k}^H \mathbf{\Phi} \mathbf{H}\mathbf{f}\sqrt{p_k}s_k + \sum\nolimits_{i=k+1}^{K} \mathbf{h}_{k}^H \mathbf{\Phi} \mathbf{H}\mathbf{f}\sqrt{p_i} s_i + n_k. \label{PCSI_Received_Signal_AfterSIC}
\end{align} 
In the following section, we explain the channel model adopted in this paper.
\subsection{Channel model}
The mmWave channel is sparse in the angular domain because of the high path-loss and absorption losses incurred at these high-frequency bands. This distinguishes it from the rich scattering environment typically encountered in lower frequency bands. To capture and describe this sparsity, we employ a widely recognized, non-geometrical stochastic model known as the extended Saleh-Valenzuela channel \cite{Intro_mmWave_4}. Using this, the RIS-user and BS-RIS channels are modeled as  
\begin{subequations}
    \begin{align}
      \mathbf{H} &= \sqrt{\frac{N_H N_V}{S_1}} \sum_{j=1}^{S_1}\alpha_{j} \mathbf{a}_N(\theta^r_{j},\phi^r_{j}) \mathbf{a}_M^H(\theta^t_{j})\label{H_ch_mod}, \\  
      \text{and}~~\mathbf{h}_{k} &= \sqrt{\frac{N_H N_V}{S_2}} \sum_{j=1}^{ S_2}\alpha_{jk}\mathbf{a}_N(\theta^t_{jk},\phi^t_{jk}),\label{ch_mod}
\end{align}
\end{subequations}
respectively, where $S_1$ and $S_2$ represent the number of paths between the BS-RIS and RIS-$k$-th user respectively, $\alpha_{j}$ and $\alpha_{jk}$ are the unit variance complex Gaussian gains along the $j$-th path between the BS-RIS, and RIS-$k$-th user, respectively. $\theta^t_{j}, \theta^r_{j}$, and $\phi^r_{j}$ are the azimuthal angle of departure from the BS, the azimuthal angle of arrival at the RIS, and the elevation angle of arrival at the RIS of the $j$-th path. Similarly, $\theta^t_{jk}$, and $\phi^t_{jk}$ are the azimuth and elevation angles of departure from the RIS to the $k$-th user. In \eqref{H_ch_mod}, $\mathbf{a}_N(\theta,\phi)$ is the uniform planar array (UPA) steering vector at the corresponding azimuth and elevation angles of length $N$ and is given by $$\mathbf{a}_N(\theta,\phi) = \frac{1}{\sqrt{N_H N_V}} \left(\mathbf{a}_{N_H} \otimes \mathbf{a}_{N_V}\right),$$ such that $$\mathbf{a}_{N_V}(\theta) =[1~~e^{-j \pi \cos(\phi) \sin(\theta)} \cdots e^{-j (N_V-1) \pi \cos(\phi)\sin(\theta)}]^T,$$ $$\mathbf{a}_{N_H}(\theta,\phi) = [1~~e^{-j \pi \cos(\theta)} \cdots e^{-j (N_H-1) \pi \cos(\theta)}]^T,$$ where $N_H$ and $N_V$ are the number of elements along the horizontal and vertical directions of the UPA. Furthermore, $\mathbf{a}_M(\theta)$ in \eqref{H_ch_mod} is the array steering vector corresponding to the uniform linear array (ULA) of length $M$ at the BS, and is defined as $$\mathbf{a}_M(\theta) = \frac{1}{\sqrt{M}}[1~~e^{-j \pi \cos(\theta)} \cdots e^{-j (M-1) \pi \cos(\theta)}]^T.$$  
\subsection{Channel estimation error model}
\label{CSI_Error_Model}
Essential information of the wireless propagation environment is encapsulated by CSI. Availability of this information at the transmitter helps it to configure the transmission for  enhancing the  link reliability and capacity \cite{26}. The process of acquiring CSI is as follows. The transmitter sends the pilot signal, and the receiver uses the received pilot signal to estimate the  CSI by applying methods such as minimum mean squared error or least squares. It then sends a quantized version of this CSI estimate to the transmitter via a feedback channel in frequency division duplex (FDD) systems or by utilizing channel reciprocity in time division duplex (TDD) systems. However, the CSI obtained by the transmitter may be erroneous due to channel estimation error, quantization error, or feedback channel delay/error \cite{25}-\cite{28}. Due to these reasons, considering that the transmitter receives perfect CSI will be an inaccurate assumption \cite{29}. Therefore, to make our analysis more realistic, we consider a widely accepted Gaussian channel error model \cite{ChannelEstimationError_Gauss} to account for the error in the CSI obtained by the transmitter that is written as 
 \begin{align}\label{CSI_Error_Channel}
\mathbf{H} =  \widehat{\mathbf{H}} +  \mathbf{\Lambda},~~\mathbf{h}_{k} =  \widehat{\mathbf{h}}_{k} +  \mathbf{\Lambda}_{k},
 \end{align}
where $\widehat{\mathbf{H}}~\text{and}~\widehat{\mathbf{h}}_{k}$ are the estimated BS-RIS and RIS-$k$-th user channels respectively, and $\mathbf{\Lambda}_{{:,i}},~\mathbf{\Lambda}_{k}$ are channel estimation error matrices whose entries are assumed to be independent and zero mean complex Gaussian with variance $\sigma^2_{\epsilon}$. By substituting (\ref{CSI_Error_Channel}) in \eqref{PCSI_Received_Signal_AfterSIC}, the signal received at the $k\text{-th}$ user after perfect SIC can be expressed as
\begin{align}
y_k &= \widehat{\mathbf{h}}_{k}^H \mathbf{\Phi} \widehat{\mathbf{H}}\mathbf{f}\sqrt{p_k} s_k + \sum_{i=k+1}^{K} \widehat{\mathbf{h}}_{k}^H \mathbf{\Phi} \widehat{\mathbf{H}}\mathbf{f} \sqrt{p_i} s_i + \mathbf{v}_k + n_k,\label{received_signal_csi}
\end{align}
where
$$\mathbf{v}_k = (\mathbf{\Lambda}_{k}^H \mathbf{\Phi} \widehat{\mathbf{H}} \mathbf{f} + \widehat{\mathbf{h}}_{k}^H \mathbf{\Phi} \mathbf{\Lambda} \mathbf{f} + \mathbf{\Lambda}_{k}^H \mathbf{\Phi} \mathbf{\Lambda}_1 \mathbf{f}) \sum\nolimits_{i=1}^{K}\sqrt{p_i}  s_i,$$ is the residual error term arising due to imperfect CSI (ICSI) estimation. 
The variance of $\mathbf{v}_k$ is given by
\begin{align}
     &\alpha_k = \mathbb{E}[\|\mathbf{v}_k\|^2] = \Big\{ \sigma^2_{\epsilon} (\|\widehat{\mathbf{h}}_{k}\|^2\|\mathbf{f}\|^2 + \|\widehat{\mathbf{H}} \mathbf{f}\|^2) + N\sigma^4_{\epsilon}\|\mathbf{f}\|^2 \Big\} \sum\limits_{i = 1}^Kp_i. \label{ICSI_Power}
\end{align}
\section{Problem Formulation}\label{Problem_Form}
In this subsection, we present the  sum rate   and energy efficiency maximization problems to obtain joint transmit beamformer $\mathbf{f}$, RIS phase shift matrix $\mathbf{\Phi}$, and power allocation $\mathbf{p}$ solutions. 
We begin by defining the signal-to-interference plus noise ratio (SINR) at the $k$-th user using \eqref{received_signal_csi} as 
\begin{align}
       \Gamma_k &= \frac{|\widehat{\mathbf{h}}_{k}^H \mathbf{\Phi} \widehat{\mathbf{H}}\mathbf{f}|^2p_k}{\sigma^2 + \sum\nolimits_{i=k+1}^{K}|\widehat{\mathbf{h}}_{k}^H \mathbf{\Phi} \widehat{\mathbf{H}}\mathbf{f}|^2p_i+\alpha_k},\label{SINR} 
\end{align}
where $\sigma^2$ is the noise power and $\alpha_k$ is the interference power due to ICSI.
Using \eqref{SINR}, the achievable rate of the $k$-th user is given as $R_k = \log_2(1+\Gamma_k)$ bits/sec. In order to ensure a certain quality of service, we consider minimum rate constraint such that the rate $R_k$ received by the $k$-th user is above $R_{k,\rm{th}}$, i.e., $R_k \geq R_{k,\rm{th}}$. This constraint can be rewritten in the form of minimum SINR constraint as $\Gamma_k \geq 2^{R_{k,\rm{th}}} - 1 = \eta_k$. We now present the sum rate  maximization problem formulation that enhances the overall capacity of the system as
\begin{subequations}\label{SR}
\begin{align}
\max_{\mathbf{p},\mathbf{f}, \boldsymbol{\psi}} \quad  &\sum_{k=1}^{K} R_k \label{SR_Obj} \\
\textrm{s.t.} \hspace{.5cm} & |\boldsymbol{\psi}_i| = 1, \forall i = 1, \cdots N,\label{SR_Phi_Con}\\ 
&\|\mathbf{f}\|^2 = 1,~~|\mathbf{f}_i| \leq 2, \forall i=1,\dots,M, \label{SR_f_con}\\
&\sum\nolimits_{k=1}^{K}p_k \leq P_s,~~p_k \geq 0,\forall k = 1,\dots. K, \label{SR_P_con} \\
&\Gamma_k \geq \eta_k,~\forall k = 1,\dots, K, \label{SR_MinRate_Con}
\end{align}
\end{subequations}
where the first constraint \eqref{SR_Phi_Con} is the unit modulus constraint of the RIS phase shift matrix to ensure its passive assumption, \eqref{SR_f_con} ensures that the beamformer is in accordance with the double-phase shifter-based analog beamforming architecture, \eqref{SR_P_con} ensures that the power allocation does not exceed the total available power $P_S$, and \eqref{SR_MinRate_Con} ensures minimum rate constraint of the users.
The optimization formulation given in \eqref{SR} is non-convex due to 1) fractional objective, 2) unit modulus constraint of the RIS phase shift matrix, and 3) minimum rate constraint. Besides, the optimization problem is difficult to solve directly because of the coupled optimization variables. To solve such a problem, we first propose a simple yet effective reformulation to bring the fractional objective into a tractable quadratic form. Next, we relax the non-convex constraints by 1) using the SDR technique to tackle the unit modulus constraint of the RIS phase shift matrix and 2) linearizing the minimum rate constraint by using the first-order Taylor series approximation. Finally, to decouple the optimization variables, we propose an alternating optimization-based algorithm that solves individual subproblems iteratively until the  sum rate   converges. We discuss all the above-mentioned details in Section \ref{SR_Section}.

In addition, we also present an optimization formulation that maximizes the energy efficiency as
\begin{subequations}\label{P_EE}
\begin{align}
\max_{\mathbf{p},\mathbf{f}, \boldsymbol{\psi}} \quad & \frac{\sum_{k=1}^{K} R_k}{\sum_{k = 1}^{K}p_k + P_c}\label{EEMax1}, \\
\textrm{s.t.} \quad &\eqref{SR_Phi_Con}-\eqref{SR_MinRate_Con}
\end{align}
\end{subequations}
where $P_c$ represents the transmitter operational power consumption and can be modeled as \cite{RIS_EE_PcMod} 
\begin{align}
    P_c =  P_{BS} + N P_{RIS},
\end{align}
where  $P_{BS}$ is the power consumed at BS and $P_{RIS}$ is the minimal power consumed by the passive elements of RIS. 
The above optimization problem is non-convex as it falls under the class of fractional programming problems (FP). To solve this problem, we transform the problem into an equivalent quadratic form and use the same strategies mentioned-above to simplify the non-convex constraints which will be explained in detail in Section \ref{EE_Section}.

\section{Problem Reformulation}
The  sum rate   and energy efficiency maximization problems fall under the category of problems classified as the \textit{concave-convex FP class} fractional programming that involve multiple ratios in the objective function. Such problems are NP-hard and thus require exponential running time \cite{FP_Exp} to arrive at a solution. Two popular methods are the Dinklebach \cite{FP_DB} and Chance Cooper \cite{FP_ChanceCoop} methods that solve fractional programming problems by decoupling the numerator and denominator. However, these techniques are more accurate in the case of \textit{single ratio} which corresponds to the single user case in our formulation. Thus, we cannot apply either of these techniques directly to objective functions involving a \textit{sum-of-function-of-ratios}, as is the case for sum rate and energy efficiency objective functions. Hence, we investigate a novel approach that is based on the \textit{quadratic transformation} method proposed in \cite{FP_QT} to decouple the numerator and denominator of the sum-of-function-of-fractions by introducing auxiliary variables. This transform will be presented in detail in the following subsection. 

\subsection{Quadratic Transformation} \label{QT_defining_section}
The quadratic transformation method is a generalization of the Dinklebach method that introduces an auxiliary variable $\mathbf{y}$ to decouple the numerator and the denominator of fractional objective function such that resulting reformulation becomes a concave problem in the optimization variable $\mathbf{x}$ and quadratic in the auxiliary variable $\mathbf{y}$. Moreover, the transformation is unique, such that the optimal value of the original and reformulated objective functions are the same. In this section, we briefly present this transformation proposed in \cite{FP_QT} for completeness. 

Let us define a multidimensional complex multi-ratio fractional programming problem as 
\begin{subequations}
       \begin{align}
       \max_{\mathbf{x}}~~~~&\sum_{k=1}^K f_k \Big(1 + \mathbf{a}_k^H(\mathbf{x}) \mathbf{B}_k^{-1}(\mathbf{x}) \mathbf{a}_k(\mathbf{x}) \Big), \label{QT_Definition_eq1_Obj} \\
        \text{s.t.}~~~~&\mathbf{x} \in \chi \subseteq \mathbb{C}^{d},
    \end{align} \label{QT_Definition_eq1}%
\end{subequations}
    where $\mathbf{a}_k(\mathbf{x}): \mathbb{C}^{n} \to \mathbb{C}^n$, $\mathbf{B}_k(\mathbf{x}): \mathbb{C}^{n} \to \mathbb{S}_{++}^n$, $f_k(\cdot)$ is a sequence of non-decreasing function, and  $n \in \mathbb{N}$. 
The goal is to simplify the ratio $\mathbf{a}_k^H(\mathbf{x}) \mathbf{B}_k^{-1}(\mathbf{x}) \mathbf{a}_k(\mathbf{x})$ in the objective function by introducing auxiliary variables $\nu_k$ such that it reduces to a tractable form as 
\begin{align*}
    2 \Re\{ {\nu}_k^{\star} \mathbf{a}_k(\mathbf{x}) \} - |{\nu}_k|^2 \mathbf{B}_k(\mathbf{x}),
\end{align*}
such that $\nu_k \in \mathbb{C},~\forall k = 1, \cdots, K$ are the auxiliary variables introduced to decouple the numerator and the denominator. 
The above term is maximized with respect to the auxiliary variable $\nu_k$ at $$\nu_k^{\rm{opt}} = \frac{\mathbf{a}_k(\mathbf{x})}{\mathbf{B}_k(\mathbf{x})},$$ such that $$2 \Re\{ ({{\nu}_k^{\rm{opt}}})^* \mathbf{a}_k(\mathbf{x}) \} - |{\nu}_k^{\rm{opt}}|^2 \mathbf{B}_k(\mathbf{x}) = \mathbf{a}_k^H(\mathbf{x}) \mathbf{B}_k^{-1}(\mathbf{x}) \mathbf{a}_k(\mathbf{x}).$$
Using this, we can write the reformulated optimization problem that is a function of the auxiliary variable as 
\begin{subequations}
       \begin{align}     
       \max_{\mathbf{x},\nu_k}~&\sum_{k=1}^K f_k \Big(1 +  2 \Re\{ {\nu}_k^* \mathbf{a}_k(\mathbf{x}) \} - |{\nu}_k|^2 \mathbf{B}_k(\mathbf{x}) \Big), \label{QT_Definition_eq1_ReformObj} \\
        \text{s.t.}~~~~&\mathbf{x} \in \chi \subseteq \mathbb{C}^{d}, \nu_k \in \mathbb{C}.
    \end{align} \label{QT_Definition_Reform_eq1}
\end{subequations} 
The above problem is convex in $\mathbf{x}$ and quadratic in $\nu_k$ and thus can be solved by iterating over $\mathbf{x}$ and $\nu_k$. More specifically, when the numerator is a concave function and the denominator is a convex function, such that the optimization problem falls under the concave-convex FP problem, the iterative algorithm converges to a stationary point as presented in \cite{FP_QT}.
We adopt this transform to solve the optimal beamforming, RIS phase shift matrix, and power allocation subproblems corresponding to both the sum rate and energy efficiency maximization problems, which will be discussed in the following subsections. 
 
It is to be noted that another widely used transform called the weighted mean square error (WMSE) reformulation \cite{WMSE_SR_Equiv} has been extensively used in the literature to solve the  sum rate maximization problem in the context of multi-antenna system design. 
For example, \cite{WMSE_1,WMSE_2,WMSE_3,WMSE_4,WMSE_5} utilizes the WMSE reformations to solve the sum-rate maximization problem under various system settings and also investigate its convergence properties.
This reformulation is a versatile technique and  has also been used as a part of other algorithms, such as Minorization-maximization, to solve non-convex optimization problems. In this paper, we  interestingly observe that WMSE and quadratic transforms are closely related. 
\begin{lemma}
The quadratic transform is an equivalent form of WMSE reformulation for sum rate maximization.
\end{lemma}
\begin{proof}
 Please refer to Appendix \ref{Appendix_QT_WMSE_equiv} for the proof.    
\end{proof}

\subsection{Sum-rate maximization}\label{SR_Section}
In this section, we present the sum rate  maximization problem of the RIS-aided mmWave-NOMA communication system that ensures minimum rates to each user. In particular, we aim to obtain the optimal beamformer $\mathbf{f}$, RIS phase shift matrix $\mathbf{\Phi}$, and power allocation $\mathbf{p}$ that maximizes SR. However, solving the sum rate  maximization problem is mathematically intractable due to 1) fractional SINR terms, 2) non-convex constraints, and 3) coupled optimization variables. Thus, we tackle the problem by reformulating the individual transmit beamformer, phase shift matrix, and power allocation subproblems into equivalent formulations using the quadratic transform discussed in Section \ref{QT_defining_section}. Finally, to obtain a unified solution, we propose an algorithm based on alternating optimization. 
We start by presenting the equivalent quadratic transformation-based reformulations of individual subproblems in the following subsections.

\subsubsection{Transmit beamformer subproblem}\label{SR_Tx_SubProblem}
For given ${\bf p}$ and $\mathbf{\Phi}$, the optimal beamformer ${\bf f}$ subproblem becomes
\begin{subequations}\label{OptiProb:SR_Tx_SubProblem}
\begin{align}
\max_{\mathbf{f}} \quad &\sum_{k=1}^{K} \log_2\Big(1 + \Gamma_k(\mathbf{f}) \Big),\label{SR_Tx_Opti_Obj}\\
\textrm{s.t.} \quad &\left\lVert\mathbf{f}\right\lVert^2 = 1, |\text{{\bf f}}_i| \leq 2; \forall i = 1,\dots M, \\ 
&\Gamma_k(\mathbf{f}) \geq \eta_k;~\forall k = 1, \cdots K.
\end{align}
\end{subequations}
In \eqref{SR_Tx_Opti_Obj}, SINR as a function of $\mathbf{f}$ can be written using \eqref{SINR} as 
\begin{align}
        \Gamma_k(\mathbf{f}) &= \frac{|\mathbf{a}_k \mathbf{f}|^2}{\sigma^2 + \mathbf{f}^H \mathbf{A}_k \mathbf{f}},\label{SINR_f}
\end{align}
where
 \begin{align*}\
    \mathbf{a}_k &= \mathbf{g}_k\sqrt{\mathbf{p}_k}, \\
    \mathbf{A}_k &= \mathbf{g}_k^H \mathbf{g}_k \sum^K_{i = k+1} \mathbf{p}_i  + \mathbf{Z}_k \sum^K_{i=1} \mathbf{p}_i, \\ 
    \mathbf{Z}_k &= \Big(\sigma^2_{\epsilon} \|\widehat{\mathbf{h}}_{k}\|^2 \mathbf{I} + \sigma^2_{\epsilon} \widehat{\mathbf{H}}^H \widehat{\mathbf{H}} + \sigma^4_{\epsilon} N \mathbf{I}\Big), \\
    \mathbf{g}_k &= \widehat{\mathbf{h}}_{k}^H \mathbf{\Phi} \widehat{\mathbf{H}},
\end{align*}
$\forall k = 1,\cdots K$. To make the problem tractable, we apply the quadratic transform discussed in Section \ref{QT_defining_section}. The equivalent representation of the $\Gamma_k(\mathbf{f})$ term after introducing auxiliary variable $y_k$ is
\begin{align}
          \widehat{\Gamma}_k(\mathbf{f},y_k)       
    &= 2\Re\{ y_k^{\star} \mathbf{a}_k \mathbf{f}\} - |y_k|^2 \Big( \sigma^2 + \mathbf{f}^H \mathbf{A}_k \mathbf{f} \Big) \label{SINR_QT_Tx}.
\end{align}
It is now very clear that the reformulated objective is concave in $\mathbf{f}$ and quadratic in $y_k$. Nonetheless, the above formulation is still non-convex due to the minimum rate constraint. This constraint can be rewritten using \eqref{SINR_f} as 
\begin{subequations}
    \begin{align*}
   \Gamma_k(\mathbf{f}) \geq \eta_k &\Rightarrow  |\mathbf{a}_k \mathbf{f}|^2 \geq  \eta_k \Big( \sigma^2 + \mathbf{f}^H \mathbf{A}_k \mathbf{f} \Big), \\
   &\Rightarrow \mathbf{f}^H \mathbf{B}_k \mathbf{f} \geq \eta_k \sigma^2, \label{MinRate_Con:SR_Tx_SubProblem}
\end{align*}
\end{subequations}
where $\mathbf{B}_k = \mathbf{a}^H_k \mathbf{a}_k - \eta_k \mathbf{A}_k$.
Further, we obtain a concave lower bound by using the Taylor series expansion at a feasible point $\mathbf{f}_{\rm{o}}$ as   
\begin{align}
    2 \Re\{\mathbf{f}_{\rm{o}}^H \mathbf{B}_k \mathbf{f}\} - \|\Tilde{\mathbf{B}}_k \mathbf{f}_{\rm{o}}\|^2 \geq \eta_k \sigma^2,
\end{align}\label{MinRate_Con_Taylor:SR_Tx_SubProblem}
where $\Tilde{\mathbf{B}}_k^H \Tilde{\mathbf{B}}_k = {\mathbf{B}_k}.$ This is also equivalent to linearizing the non-convex minimum rate constraint at any feasible point $\mathbf{f}_{\rm{o}}$. This makes the subproblem convex in $\mathbf{f}$ and can be solved for the given linearized constraint around $\mathbf{f}_{\rm{o}}$ which can be updated successively. Finally, the optimization problem to obtain the transmit beamformer is 
\begin{subequations}\label{Reformulated:OptiProb:SR_Tx_SubProblem}
\begin{align}
\max_{\mathbf{f},~y_k} \quad &\sum_{k=1}^{K} \log_2\Big(1 + \widehat{\Gamma}(\mathbf{f},y_k) \Big)\label{ReformObj:SR_Tx_SubProblem},\\
\textrm{s.t.} \quad &\left\lVert\mathbf{f}\right\lVert^2 = 1, |\text{{\bf f}}_i| \leq 2; \forall i = 1,\dots M,\\ 
&2 \Re\{\mathbf{f_{\rm{o}}}^H \mathbf{B}_k \mathbf{f}\} - \|\Tilde{\mathbf{B}}_k \mathbf{f_{\rm{o}}}\|^2 \geq \eta_k \sigma^2;\forall k = 1, \cdots K, \\
&y_k \in \mathbb{C}; \forall k = 1,\cdots K.
\end{align}
\end{subequations}
It is worth noting that the reformulated transmit beamformer subproblem is now a function of the introduced auxiliary variable $y_k$ whose optimal value for a given $\mathbf{f}$ can be obtained as  
\begin{align} \label{Reformulated:yopt:SR_Tx_SubProblem}
    y_k^{\rm{opt}} = \frac{\mathbf{a}_k \mathbf{f}}{\sigma^2 + \mathbf{f}^H \mathbf{A}_k \mathbf{f}}.
\end{align} 
Note that substituting \eqref{Reformulated:yopt:SR_Tx_SubProblem} in \eqref{SINR_QT_Tx} results in \eqref{SINR}. Through this, we can ensure that the maximum objective value of the quadratic transform-based reformulation and the original problem is the same. The reformulated problem in \eqref{Reformulated:OptiProb:SR_Tx_SubProblem} is now convex, which can be solved using standard solvers such as CVX for $\mathbf{f}$. Furthermore, in order to obtain the transmit beamformer, we iteratively solve \eqref{Reformulated:OptiProb:SR_Tx_SubProblem} for $\mathbf{f}$ and \eqref{Reformulated:yopt:SR_Tx_SubProblem} for $y_k^{\rm{opt}}$ untill \eqref{ReformObj:SR_Tx_SubProblem} converges.   
\subsubsection{RIS phase shift matrix subproblem}\label{SR_Psi_SubProblem}
For given ${\bf p}$ and ${\bf f}$, the phase shift matrix $\mathbf{\Phi} = \rm{diag}(\boldsymbol{\psi})$ subproblem is 
\begin{subequations}\label{OriginalOpti:SR_Psi_SubProblem}
\begin{align}
\max_{\boldsymbol{\psi}} \quad &\sum_{k=1}^{K} \log_2\Big(1 + \Gamma_k(\boldsymbol{\psi}) \Big)\label{SR_Psi_Opti_Obj},\\
\textrm{s.t.} \quad &|\boldsymbol{\psi}_i| = 1; \forall i = 1, \dots N,\\ 
&\Gamma_k \geq \eta_k;~\forall k = 1, \cdots K.
\end{align}
\end{subequations}
In \eqref{SR_Psi_Opti_Obj}, SINR as a function of $\boldsymbol{\psi}$ can be written using \eqref{SINR} as 
\begin{align}
    \Gamma_k(\boldsymbol{\psi}) &= \frac{|\boldsymbol{\psi}^H \mathbf{E}_k^* \mathbf{f}^* \sqrt{\mathbf{p}_k}|^2}{\sigma^2 + |\boldsymbol{\psi}^H \mathbf{E}_k^* \mathbf{f}^*|^2 \sum_{i=k+1}^K \mathbf{p}_i + \alpha_k \sum_{i=1}^K \mathbf{p}_i} \label{SINR_Psi}  
\end{align}
where $\mathbf{E}_k = \rm{diag}({\widehat{\mathbf{h}}_{k}}^*)\widehat{\mathbf{H}}$.
Similar to the transmit beamformer subproblem, we apply the quadratic transform to make the objective function tractable in terms of $\boldsymbol{\psi}$. The equivalent representation of $\Gamma_k(\boldsymbol{\psi})$ after applying the quadratic transform is
    \begin{align}
    \Tilde{\Gamma}_k(\boldsymbol{\psi},\nu_k) 
    =&~~2\Re\{\nu_k^{*} \boldsymbol{\psi}^H \mathbf{H}_k^* \mathbf{f}^* \sqrt{\mathbf{p}_k}\} - |\nu_k|^2 (\sigma^2 + \alpha_k\sum_{i=1}^K \mathbf{p}_i) \nonumber \\ &~~ - |\nu_k|^2 {\boldsymbol{\psi}}^H \Big( \mathbf{H}_k^* \mathbf{f}^* \mathbf{f}^T \mathbf{H}_k^T \sum_{i=k+1}^K \mathbf{p}_i \Big) {\boldsymbol{\psi}}, \label{SR_QT_Psi_Reform}
\end{align}
where $\nu_k$ is an auxiliary variable. 
This reformulated objective can further be simplified as  
    \begin{align}
    \Tilde{\Gamma}_k(\boldsymbol{\psi},\nu_k) 
    = &~~\Tilde{\boldsymbol{\psi}}^H \mathbf{C}_k \Tilde{\boldsymbol{\psi}} - c_k,\label{SR_Psi_SubP_SINR}
\end{align}
where 
\begin{align*}
    \mathbf{C}_k &= \begin{bmatrix}
    -|\nu_k|^2 \mathbf{H}_k^* \mathbf{f}^* \mathbf{f}^T \mathbf{H}_k^T \sum_{i=k+1}^K \mathbf{p}_i & \nu_k^{*} \mathbf{H}_k^* \mathbf{f}^* \sqrt{\mathbf{p}_k} \\
    \sqrt{\mathbf{p}_k} \mathbf{f}^T \mathbf{H}_k^T \nu_k & 0 
\end{bmatrix}, \\
c_k &= |\nu_k|^2 (\sigma^2 + \alpha_k\sum_{i=1}^K \mathbf{p}_i), \\
\Tilde{\boldsymbol{\psi}} &= [\boldsymbol{\psi}~~~1]^T.
\end{align*} 
Further, we rewrite the minimum rate constraint using \eqref{SINR_Psi} as follows
\begin{align*}
\Gamma_k(\boldsymbol{\psi}) \geq \eta_k &\Rightarrow |\boldsymbol{\psi}^H \mathbf{H}_k^* \mathbf{f}^*|^2 \Big( \mathbf{p}_k - \eta_k \sum^K_{i=k+1} \mathbf{p}_i \Big) \geq \\&\hspace{1cm}\eta_k (\sigma^2 + \alpha_k \sum_{i=1}^K \mathbf{p}_i), \\
 &\Rightarrow \Tilde{\boldsymbol{\psi}}^H \mathbf{D}_k \Tilde{\boldsymbol{\psi}} \geq d_k
\end{align*}
where $\mathbf{D}_k = \begin{bmatrix}
    \mathbf{H}_k^* \mathbf{f}^* \mathbf{f}^T \mathbf{H}_k^T & \mathbf{0} \\
    \mathbf{0}^T & 0 
\end{bmatrix}$, $d_k = \frac{\eta_k (\sigma^2 + \alpha_k \sum_{i=1}^K \mathbf{p}_i)}{\mathbf{p}_k - \eta_k \sum^K_{i=k+1} \mathbf{p}_i}$, and $\mathbf{0}$ is a zero vector of length $N$. 

The above representation of the rate constraint makes the optimization problem non-convex. Furthermore, the unit modulus constraint also makes the RIS phase shift subproblem non-convex. Thus, to solve such an optimization problem, we apply the SDR technique which we discuss next. We begin by defining $$\mathbf{\Psi} =\Tilde{\boldsymbol{\psi}} \Tilde{\boldsymbol{\psi}}^H.$$
The reformulated SINR given in \eqref{SR_Psi_SubP_SINR} in terms of the newly defined variable $\mathbf{\Psi}$ can be written as 
$$\Tilde{\Gamma}_k(\boldsymbol{\psi},\nu_k) = \rm{trace}(\mathbf{\Psi} \mathbf{C}_k) - c_k.$$
Further the minimum rate constraint can also be written as $$\textrm{trace}(\mathbf{\Psi} \mathbf{D}_k) \geq d_k.$$
Thus, the optimization problem given in \eqref{OriginalOpti:SR_Psi_SubProblem} can now be reformulated as a semidefinite programming problem as 
\begin{subequations}
\begin{align}
\max_{\mathbf{\Psi},~\nu_k} \quad &\sum_{k=1}^{K} \log_2\Big(1 - c_k + \rm{trace}(\mathbf{\Psi} \mathbf{C}_k) \Big),\\
\textrm{s.t.} \quad & \textrm{diag}(\mathbf{\Psi}) = 1,~~\textrm{rank}\left( \mathbf{\Psi} \right) = 1,\\ 
& \textrm{trace}(\mathbf{\Psi} \mathbf{D}_k) \geq d_k;~\forall k = 1, \cdots K,\\
&\nu_k \in \mathbb{C}; \forall k = 1,\cdots K,
\end{align}
\end{subequations}
where the constraint $\textrm{rank}\left( \mathbf{\Psi} \right) = 1$ ensures $\mathbf{\Psi} = \Tilde{\boldsymbol{\psi}}\Tilde{\boldsymbol{\psi}}^H$. By applying SDR, we drop the non-convex rank 1 constraint to solve the RIS phase shift matrix subproblem as 
\begin{subequations}\label{ReformOpti:SR_Psi_SubProblem}
\begin{align}
\max_{\mathbf{\Psi},~\nu_k} \quad &\sum_{k=1}^{K} \log_2\Big(1 - c_k + \rm{trace}(\mathbf{\Psi} \mathbf{C}_k) \Big) \label{ReformObj:SR_Psi_SubProblem} ,\\
\textrm{s.t.} \quad &\textrm{diag}(\mathbf{\Psi}) = 1,\\ 
&\textrm{trace}(\mathbf{\Psi} \Tilde{\mathbf{D}}_k) \geq d_k;~\forall k = 1, \cdots K,\\
&\nu_k \in \mathbb{C}; \forall k = 1,\cdots K.
\end{align}
\end{subequations}
This reformulated problem is a function of the auxiliary variable $\nu_k$ whose optimal solution can be obtained as 
\begin{align}\label{Reformulated:nuopt:SR_Psi_SubProblem}
    \nu_k^{\rm{opt}} = \frac{\boldsymbol{\psi}^H \mathbf{H}_k \mathbf{f} \sqrt{\mathbf{p}_k}}{\sigma^2 + |\boldsymbol{\psi}^H \mathbf{H}_k \mathbf{f}|^2 \sum_{i=k+1}^K \mathbf{p}_i + \alpha_k \sum_{i=1}^K \mathbf{p}_i};~\forall k.
\end{align}
It is worth noting that upon substituting \eqref{Reformulated:nuopt:SR_Psi_SubProblem} in \eqref{SR_QT_Psi_Reform}, we obtain \eqref{SINR} which ensures that the quadratic transform-based reformulation and the original problem has the same maximum objective value. The reformulated problem in \eqref{ReformOpti:SR_Psi_SubProblem} is now convex, which can be solved using standard solvers such as CVX for $\boldsymbol{\psi}$. Furthermore, in order to obtain the RIS phase shift vector, we iteratively solve \eqref{ReformOpti:SR_Psi_SubProblem} for $\boldsymbol{\psi}$ and \eqref{Reformulated:nuopt:SR_Psi_SubProblem} for $\nu_k^{\rm{opt}}$ untill \eqref{ReformObj:SR_Psi_SubProblem} converges.
\subsubsection{Power allocation subproblem}\label{SR_PA_SubProblem}
 For given $\boldsymbol{\psi}$ and $\mathbf{f}$, the power allocation subproblem is defined as 
\begin{subequations}\label{OriginalOpti:SR_PA_SubProblem}
    \begin{align}
        \max_{\mathbf{p}} \quad &\sum_{k=1}^{K} \log_2\Big(1 + \Gamma_k(\mathbf{p}) \Big)\label{SR_P_Opti_Obj},\\
\textrm{s.t.} \quad &\sum_{k=1}^K \mathbf{p}_k \leq P_s, \mathbf{p} \geq \mathbf{0}, \\ &\Gamma_k \geq \eta_k;~\forall k = 1, \cdots K.
    \end{align}
\end{subequations}
In \eqref{SR_P_Opti_Obj}, SINR as a function of $\mathbf{p}$ can be written using \eqref{SINR} as \begin{align}
    \Gamma_k(\mathbf{p})  = \frac{\mathbf{p}_k}{\mathbf{a}_k + \sum_{i=k+1}^K \mathbf{p}_i+\mathbf{b}_k \sum_{i=1}^K \mathbf{p}_i}, \label{SR_SINR_p}
\end{align}
where $\mathbf{a}_k = \frac{\sigma^2}{|\mathbf{g}_k\mathbf{f}|^2}$, $\mathbf{b}_k = \frac{\alpha_k}{|\mathbf{g}_k\mathbf{f}|^2}$. 
After applying the quadratic transformation, the equivalent representation of SINR is 
\begin{align}
    \Bar{\Gamma}_k(\boldsymbol{\mathbf{p}},x_k) 
    =&~~2 x_k \sqrt{\mathbf{p}_k} - x_k^2 \Big( \mathbf{a}_k + \sum_{i=k+1}^K \mathbf{p}_i + \mathbf{b}_k \sum_{i=1}^K \mathbf{p}_i \Big). \label{SINR_p_SR_reformQT}
\end{align}
where $x_k$ is a auxiliary variable.
Further, we rewrite the minimum rate constraint as follows 
\begin{align}
        \Gamma_k(\boldsymbol{\mathbf{p}}) \geq \eta_k \Rightarrow
        \eta_k\sum_{i=k+1}^K \mathbf{p}_i - \mathbf{p}_k + \eta_k (\mathbf{a}_k + \mathbf{b}_k \sum_{i=1}^K \mathbf{p}_i ) \leq 0 \label{SR_P_Rmin}.
\end{align}
Finally, using this and \eqref{SINR_p_SR_reformQT}, the power allocation subproblem can be reformulated as 
\begin{subequations}\label{ReformOpti:SR_PA_SubProblem}
    \begin{align}
        \max_{\mathbf{p},~~x_k} \quad &\sum_{k=1}^{K} \log_2\Big(1 + \Bar{\Gamma}_k(\boldsymbol{\mathbf{p}},x_k) \Big)\label{ReformObj:SR_PA_SubProblem},\\
\textrm{s.t.} \quad &\sum_{k=1}^K \mathbf{p}_k \leq P_s, \mathbf{p} \geq \mathbf{0}, \\ &\eta_k\sum_{i=k+1}^K \mathbf{p}_i - \mathbf{p}_k + \eta_k (\mathbf{a}_k + \mathbf{b}_k \sum_{i=1}^K \mathbf{p}_i )\leq 0;~\forall k, \label{EE_PA_MinRate}\\
&x_k \in \mathbb{R}; \forall k = 1,\cdots K.
    \end{align}
\end{subequations}
The reformulated optimization problem is a function of the auxiliary variable $x_k$ whose optimal value can be obtained as 
\begin{align}\label{Reformulated:xopt:SR_PA_SubProblem}
    x_k^{\rm{opt}} = \frac{\sqrt{\mathbf{p}_k}}{\mathbf{a}_k + \sum_{i=k+1}^K \mathbf{p}_i + \mathbf{b}_k}.
\end{align}
Similar to the other subproblems, it is worth noting that substituting \eqref{Reformulated:xopt:SR_PA_SubProblem} in \eqref{SINR_p_SR_reformQT} results in \eqref{SINR} that ensures the maximum objective value of the quadratic transform-based reformulation and the original problem is the same. This problem is concave in $\mathbf{p}$ and quadratic in $x_k;\forall k = 1,\cdots K$, and thus, can be solved by using a standard solver such as CVX for $\mathbf{p}$ and using \eqref{Reformulated:xopt:SR_PA_SubProblem} for $x_k$. Finally, the optimal power allocation for a given $\mathbf{f}$ and $\boldsymbol{\psi}$ can be obtained by iterating over the solutions of $\mathbf{p}$ and $x_k$, untill \eqref{ReformObj:SR_PA_SubProblem} converges.

\subsubsection{Quadratic transform-based algorithm for  sum rate   maximization}
So far, we have presented techniques to reduce individual subproblems into convex forms to obtain the transmit beamformer $\mathbf{f}$, RIS phase shift vector $\boldsymbol{\psi}$, and power allocation $\mathbf{p}$ solutions. However, in order to obtain a unified solution that maximizes the  sum rate   maximization problem given in \eqref{SR}, we propose Algorithm \ref{Alg_QT_SR} that iteratively solves the reformulated optimization problems given in \eqref{Reformulated:OptiProb:SR_Tx_SubProblem}, \eqref{ReformOpti:SR_Psi_SubProblem}, and \eqref{ReformOpti:SR_PA_SubProblem} for $\mathbf{f}$, $\boldsymbol{\psi}$, and $\mathbf{p}$, respectively. Moreover, it is important to note that all the subproblems are based on the quadratic transform, which needs to be solved iteratively over the optimization variable $\mathbf{f}$/$\boldsymbol{\psi}$/$\mathbf{p}$ and its corresponding auxiliary variable $y_k$/$\nu_k$/$x_k$. However, it can be
 verified from \eqref{Reformulated:nuopt:SR_Psi_SubProblem} and \eqref{Reformulated:xopt:SR_PA_SubProblem} that the
 the auxiliary variables $y_k$ and $\nu_k$ are same for the beamformer $\mathbf{f}$ and RIS phase shift $\boldsymbol{\psi}$ subproblems. Thus, we use the same auxiliary variable to solve them. On the other hand, the auxiliary variable associated with power allocation subproblem is in real space and hence we update it separately. We summarize this framework in Algorithm \ref{Alg_QT_SR}.  
\begin{algorithm}[h!]\label{Alg_QT_SR}
\KwInput{$P_s$, $\widehat{\mathbf{H}},~\widehat{\mathbf{H}}_2$.}
\KwInit{$\mathbf{f}$, $\boldsymbol{\psi}$, $\mathbf{p},i=0$}
\SetKwRepeat{Repeat}{Repeat}{Untill: \eqref{SR} converges}
\Repeat{}{
$i \to i + 1$ \\
Update $\nu_k;~\forall k = 1, \cdots, K$ using \eqref{Reformulated:nuopt:SR_Psi_SubProblem}\\
Update $\boldsymbol{\psi}$ by solving 
\eqref{ReformOpti:SR_Psi_SubProblem} \\
Set $\mathbf{f}_{\rm{o}} = \mathbf{f}$ and update $\mathbf{f}$ by solving  
\eqref{Reformulated:OptiProb:SR_Tx_SubProblem} \\
\SetKwRepeat{Repeat}{Repeat}{Untill:}
\Repeat{Untill \eqref{ReformObj:SR_PA_SubProblem} converges}
{
Update $x_k;~\forall k$ using \eqref{Reformulated:xopt:SR_PA_SubProblem}\\
Update $\mathbf{p}$ by solving \eqref{ReformOpti:SR_PA_SubProblem}
}
}
\caption{Quadratic transform-based algorithm for sum rate maximization}
\end{algorithm}

\subsection{Energy Efficiency Maximization}\label{EE_Section} 
In this section, we aim to maximize the energy efficiency of the proposed system model while ensuring the minimum transmission rate to each user. In particular, we aim to obtain optimal beamformer $\mathbf{f}$, RIS phase shift matrix $\boldsymbol{\psi}$, and power allocation $\mathbf{p}$ that maximize the energy efficiency formulation presented in \eqref{P_EE}.
It is important to recall here that energy efficiency and sum rate  are equivalent problems with respect to the RIS phase shift $\boldsymbol{\psi}$ and transmit beamformer subproblems $\mathbf{f}$, and only differ in terms of the power allocation $\mathbf{p}$ subproblem. Thus, we solve the power allocation subproblem $\mathbf{p}$ using quadratic transform-based framework, and utilize the solutions for $\mathbf{f}$ and $\boldsymbol{\psi}$ given in sections \ref{SR_Tx_SubProblem} and \ref{SR_Psi_SubProblem}, respectively in the following subsection.  
\subsubsection{Power allocation subproblem} \label{EE_PA_Subproblem_Section}
Energy efficiency is the ratio of sum rate and total power consumption. This makes the numerator of the objective of energy efficiency problem multi-ratio as seen in \eqref{OriginalOpti:SR_PA_SubProblem}. Moreover, the power allocation sub-problem maximizing the overall energy efficiency becomes a single ratio problem with a multi-ratio numerator and a linear denominator in $\mathbf{p}$. In order to make such a problem tractable, we first apply quadratic transform to the sum rate, i.e., to the numerator, as done in Section \ref{SR_PA_SubProblem}, and then apply the quadratic transform again to the overall energy efficiency ratio. We begin by writing the power allocation subproblem maximizing the energy efficiency \eqref{P_EE} for given $\mathbf{f}$ and $\boldsymbol{\psi}$ as 
\begin{subequations}
\begin{align}
\max_{\mathbf{p}} \quad &\frac{\sum_{k=1}^K \log_2 ( 1 + \Gamma_k(\mathbf{p}) )}{\sum^K_{k=1}{\bf{p}}_k + P_c},\label{EE_PA_Obj} \\
\textrm{s.t.} \quad &\sum_{k=1}^{K}{\bf{p}}_k \leq P_s~,~{\bf{p}}_k \geq 0;~\forall k = 1, \cdots K\label{EE_PA_Con1}, \\
&\Gamma_k \geq \eta_k;~\forall k = 1, \cdots K\label{EE_PA_Con2}.
\end{align}\label{EE_PA}
\end{subequations}
In \eqref{EE_PA_Obj}, $\Gamma_k(\mathbf{p})$ as a function of $\mathbf{p}$ can be written as given in \eqref{SR_SINR_p}. 
However, as mentioned in Section \ref{SR_PA_SubProblem}, $\Gamma_k(\mathbf{p})$ is a fraction in $\mathbf{p}$. Hence, by applying the quadratic transform to the numerator of \eqref{EE_PA_Obj}, we can rewrite the power allocation subproblem as%
    \begin{align*}
  \max_{\mathbf{p},x_k} \quad  &\frac{\sum_{k=1}^K \log_2 ( 1 + \Bar{\Gamma}_k(\mathbf{p}) )}{\sum^K_{k=1}{\bf{p}}_k + P_c} \\
  \textrm{s.t.} \quad & \eqref{EE_PA_Con1}-\eqref{EE_PA_Con2},\\
  &x_k \in \mathbb{R}; \forall k = 1, \cdots, K,
\end{align*}
where $\Bar{\Gamma}_k(\mathbf{p})$ is given in \eqref{SINR_p_SR_reformQT} and $x_k$ is an auxiliary variable. Nonetheless, the above reformulated subproblem still falls under concave-convex FP since the numerator and denominator of its objective function are concave and linear in $\mathbf{p}$.  Thus, we apply quadratic transform again and rewrite the above problem as%
\begin{subequations}
 \begin{align}
\max_{\mathbf{p},x_k,z} \quad &  2z\sqrt{\sum_{k=1}^K \log_2\Big(  1 + \Bar{\Gamma}_k(\boldsymbol{\mathbf{p}},x_k) \Big)} - z^2 \Big(\sum_{k=1}^K{\bf p}_k + P_c\Big) \label{EE_PA_Obj_Final} \\
\textrm{s.t.} \quad &\sum_{k=1}^{K}{\bf{p}}_k \leq P_s~,~{\bf{p}}_k \geq 0;~\forall k = 1, \cdots K \label{EE_PA_Con1_Final}, \\
\eta_k\sum_{i=k+1}^K& \mathbf{p}_i - \mathbf{p}_k + \eta_k (\mathbf{a}_k + \mathbf{b}_k \sum_{i=1}^K \mathbf{p}_i )\leq 0;~\forall k,\label{EE_PA_Con2_Final}\\
& x_k \in \mathbb{R}; \forall k = 1,\cdots K, z \in \mathbb{R} \label{EE_PA_Con3_Final},
\end{align}\label{EE_PA_Opti_Final} %
\end{subequations}where $z$ is an auxiliary variable and \eqref{EE_PA_Con2_Final} represents the minimum rate constraint that follows from \eqref{EE_PA_MinRate}. 
The above reformulated optimization problem is convex in $\mathbf{p}$ with  correponding optimal auxiliary variables 
\begin{align}
        x_k^{\rm{opt}} = \frac{\sqrt{\mathbf{p}_k}}{\mathbf{a}_k + \sum_{i=k+1}^K \mathbf{p}_i + \mathbf{b}_k},\label{Reformulated:xopt:EE_PA_SubProblem}
\end{align}
and  
\begin{align}
    z^{\rm{opt}} = \frac{\sqrt{\begin{aligned}   
    &\sum_{k = 1}^K \log_2 \Big( 1 + \Bar{\Gamma}_k(\boldsymbol{\mathbf{p}},x_k) \Big)
    \end{aligned}}}{\sum_{k=1}^K {\bf p}_k + P_c} \label{Reformulated:zopt:EE_PA_SubProblem},
\end{align}
respectively. 
These optimal values of auxiliary variables ensure that the objective is equivalent to the original FP problem. Hence, to obtain the power allocation solution $\mathbf{p}$, we iteratively solve \eqref{EE_PA_Opti_Final} for $\mathbf{p}$, \eqref{Reformulated:xopt:EE_PA_SubProblem} for $x_k^{\rm{opt}}$, and \eqref{Reformulated:zopt:EE_PA_SubProblem} for $z^{\rm{opt}}$, until \eqref{EE_PA_Obj_Final} converges. 

\subsubsection{Quadratic transform-based algorithm for energy efficiency maximization} 
In this section, we present an iterative framework for energy efficiency maximization problem given in \eqref{P_EE} that iterates over  transmit beamforming $\mathbf{f}$, RIS phase shift $\boldsymbol{\psi}$ and power allocation $\mathbf{p}$ subproblems given in 
\eqref{Reformulated:OptiProb:SR_Tx_SubProblem}, \eqref{ReformOpti:SR_Psi_SubProblem}, and \eqref{EE_PA_Opti_Final}, respectively. These subproblems require solving for optimization variables $\mathbf{f}$/$\boldsymbol{\psi}$/$\mathbf{p}$ and their corresponding auxiliary variables $y_k$/$\nu_k$/$\{x_k,z\}$.  Similar to the  sum rate   maximization framework, since auxiliary variables of the transmit beamformer $\mathbf{f}$ and the RIS phase shift vector $\boldsymbol{\psi}$ provide the same maximum objective value, we use one common auxiliary variable to update both the subproblems. On the other hand, the power allocation subproblem has two auxiliary variables that need to be updated as shown in \eqref{Reformulated:xopt:EE_PA_SubProblem} and \eqref{Reformulated:zopt:EE_PA_SubProblem}. Algorithm \ref{Alg_QT_EE} summarizes the step for energy maximization problem.
\begin{algorithm}[h!]\label{Alg_QT_EE}
\KwInput{$P_s$, $\widehat{\mathbf{H}},~\widehat{\mathbf{H}}_2$.}
\KwInit{$\mathbf{f}$ , $\boldsymbol{\psi}$, $\mathbf{p},i=0$}
\SetKwRepeat{Repeat}{Repeat}{Untill: \eqref{P_EE} converges}
\Repeat{}{
$i \to i + 1$ \\
Update $\nu_k;~\forall k = 1, \cdots K$ using \eqref{Reformulated:nuopt:SR_Psi_SubProblem}\\
Update $\boldsymbol{\psi}$ by solving 
\eqref{ReformOpti:SR_Psi_SubProblem} \\
Set $\mathbf{f}_{\rm{o}} = \mathbf{f}$ and update $\mathbf{f}$ by solving \eqref{Reformulated:OptiProb:SR_Tx_SubProblem}\\
\SetKwRepeat{Repeat}{Repeat}{} 
\Repeat{Untill \eqref{EE_PA_Obj_Final} converges}
{
Update $x_k;~\forall k$ using \eqref{Reformulated:xopt:EE_PA_SubProblem}\\
Update $z$ using \eqref{Reformulated:zopt:EE_PA_SubProblem}\\
Update $\mathbf{p}$ by solving \eqref{EE_PA_Opti_Final}
}
}
\caption{Quadratic transform-based algorithm for energy efficiency maximization}
\end{algorithm}
\section{Computational Complexity of Proposed Algorithms }
We solve the transmit beamforming and power allocation subproblems using standard solvers that model the quadratic problem as a second-order cone program and employ the interior point algorithm. In each iteration of this algorithm, a linear system of equations is solved with a complexity of $\mathcal{O}\left(n^3\right)$. The total number of iterations required to achieve the optimal solution is $\mathcal{O}\left(\sqrt{n^{3.5}} \ln{\epsilon^{-1}}\right)$, where $\epsilon$ represents the duality gap, and $n$ is the cone dimension \cite{30}, \cite{48}.  
Subsequently, the SDR technique is employed to tackle the RIS phase shift subproblem, resulting in a complexity of $\mathcal{O}(N^{9/2}\log(1/\varrho))$ \cite{SDR}, where $\varrho$ denotes the solution accuracy. Consequently, upon comparing the complexities associated with these subproblems, we conclude that both the algorithms exhibit a complexity of $\mathcal{O}(N^{9/2}\log(1/\varrho))$, which are imposed by the use of SDR. Finally, we demonstrate the convergence of the proposed algorithms through numerical simulations.

\section{Numerical Results and Discussion}
In this section, we first compare the achievable sum rate  and energy efficiency performances of RIS-NOMA-aided fully analog architecture and fully digital architecture. Next, we present the performance of the proposed RIS-aided mmWave-NOMA system for various system parameter configurations. For the numerical analysis, we set the system parameters as the power available at the BS  $P_{s} = 40~\rm{dB}$, the number of BS antennas $M = 16$, the number of RIS elements $N = 64$, the number of users $K = 4$, minimum rate requirements $R_{k,\rm{th}} = 0.3~\rm{bps/Hz};\forall k$, azimuth and elevation angles of departure and arrival $\theta_j^t, \theta_j^r, \phi_j^r, \theta_{jk}^t, \phi_{jk}^t \sim \mathcal{U}[-\pi/2, \pi/2];\forall j,k$, and the number of paths $S_1 = S_2 = 3$, unless mentioned otherwise. The mean sum rate  and energy efficiency performances presented here have been averaged over 100 channel realizations through Monte Carlo simulation.

\begin{figure}[h!]
    \centering
    \includegraphics[width=.6\columnwidth]{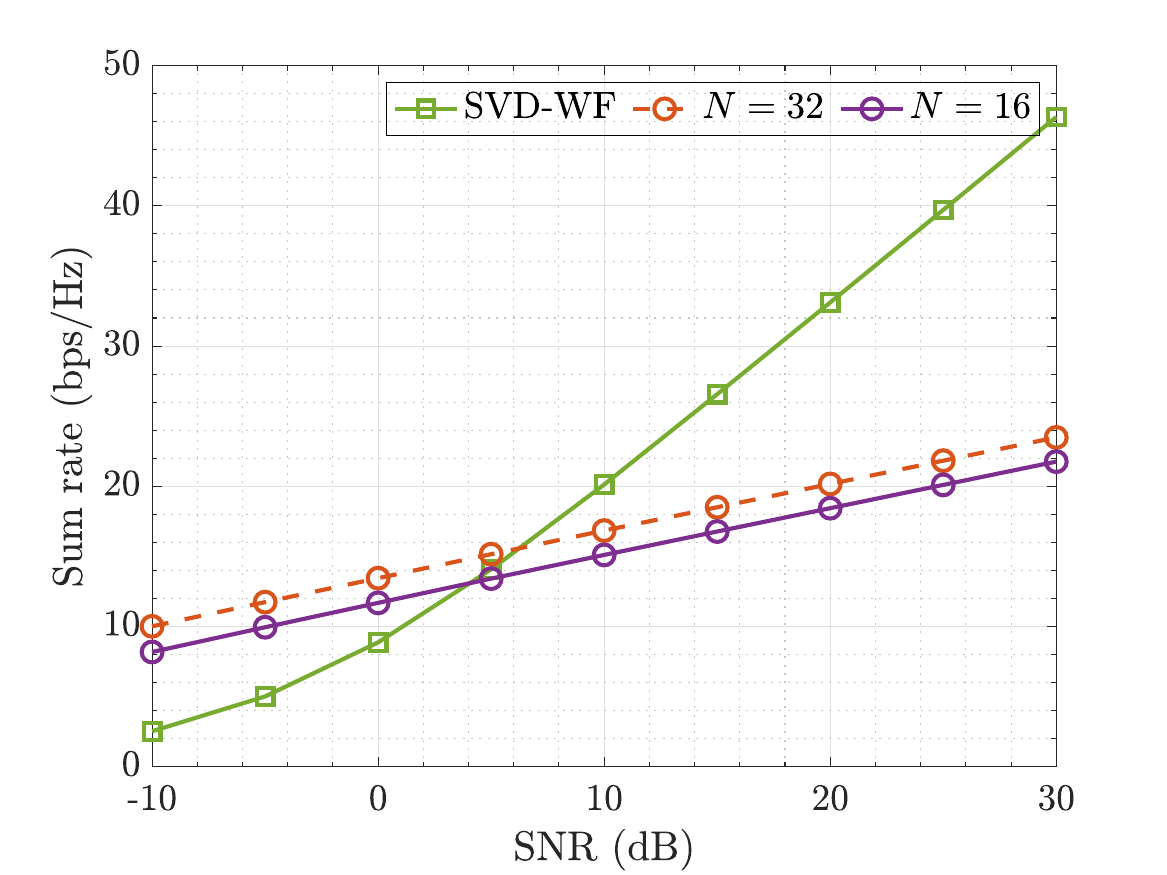}
    \caption{Sum rate performance comparison of the proposed RIS-NOMA-aided fully analog architecture with a fully digital (SVD-WF) architecture under PCSI}
    \label{SR_SNR_SVD_Comp}
\end{figure}

Figure \ref{SR_SNR_SVD_Comp} shows the sum rate  performance of the proposed system with the fully analog architecture in comparison to the fully digital MIMO system with singular value decomposition (SVD)-based precoding scheme along with optimal power allocation using the water filling algorithm referred to as SVD-WF for various values of $N$ when PCSI is available. It is worth noting that SVD-WF is a capacity-achieving method \cite{Vishw_Tse} and hence can be used as a benchmark scheme for the fully digital scheme for comparative analysis of the proposed scheme for FA. It can be observed that the proposed system model outperforms the SVD-based fully digital system  in the low SNR region. This is because of the improvement in the receive SNR at the user due to RIS array gain, which compensates for the low SNR.  
However, the performance of SVD-WF increases at a higher rate compared to the proposed scheme because the SVD-based precoder enables parallel non-interferencing channels. On the contrary, as NOMA performance is limited by inter-user interference, the performance of the proposed framework is smaller in comparison to the SVD-WF scheme at high SNR. Nonetheless, it can also be observed that the performance of the proposed system at high SNR increases towards the SVD-WF scheme with $N$. Therefore, by introducing fully analog with RIS-NOMA appropriately, we can achieve sum rate  performance comparable to the fully digital architecture while keeping the system complexity low, especially at low SNR.
\begin{figure}[h!]
    \centering
    \includegraphics[width=.6\columnwidth]{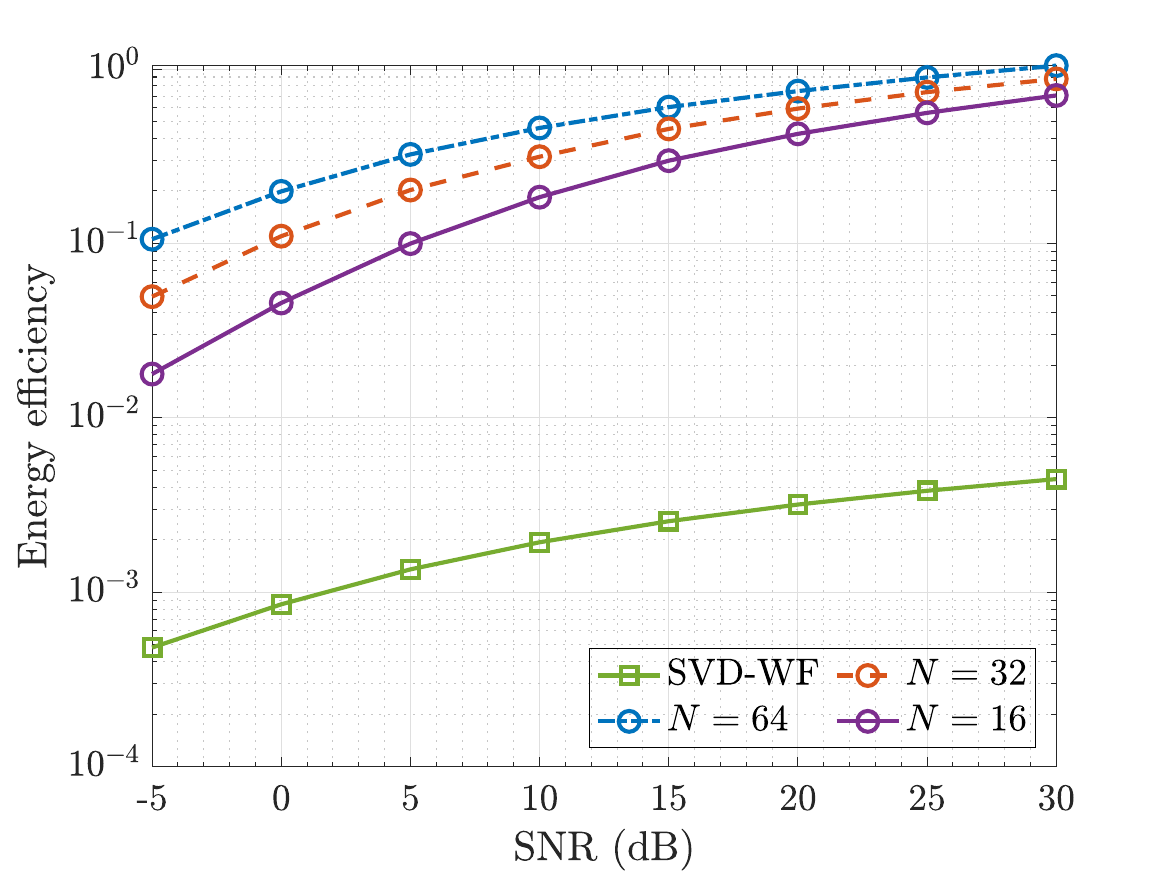}
    \caption{Energy efficiency performance comparison of the proposed RIS-NOMA-aided fully analog architecture with a fully digital (SVD-WF) architecture under PCSI}
    \label{EE_SNR_SVD_Comp}  
\end{figure}

Figure \ref{EE_SNR_SVD_Comp} shows the energy efficiency performance of the proposed RIS-NOMA-aided fully analog architecture compared to SVD-WF. 
It is important to recall here that energy efficiency is the ratio of sum rate and total power consumption $P_{tot}$ where the total power consumption depends on the circuit power consumption $P_c$ and transmission power allocated to users, i.e., $P_{tot} = \sum_{k=1}^K \mathbf{p}_k + P_c,$ such that 
$P_c$ encompasses both the power consumed at the BS and the power required for reconfiguring the passive RIS elements. Hence, we model $P_c$ for 1) the proposed RIS-NOMA-aided analog architecture as $P_c = P_{\rm{BS}}^\prime + P_{\rm{RF}} + N P_{\rm{RIS}}$ and 2) the fully digital architecture as $P_c = P_{\rm{BS}}^\prime + M P_{\rm{RF}}$, where $P_{\rm{RF}}$ is power consumption per RF chain, $P_{\rm{BS}}^\prime$ is the residual circuit power consumption at the BS, and $P_{\rm{RIS}}$ is power consumption per RIS element. 
For numerical analysis, we assume typical values for the components of  $P_c$ from \cite{RIS_EE_PcMod} and \cite{EE_SVD_PcMod}. 
Figure \ref{EE_SNR_SVD_Comp} shows that the energy efficiency performance of the proposed low complexity system is significantly better in comparison to the fully digital scheme. This is due to the use of a single RF chain in the proposed system consuming significantly lower power compared to its counterpart in the fully digital system.  

\begin{figure}[h!]
    \centering
    \includegraphics[width=.6\columnwidth]{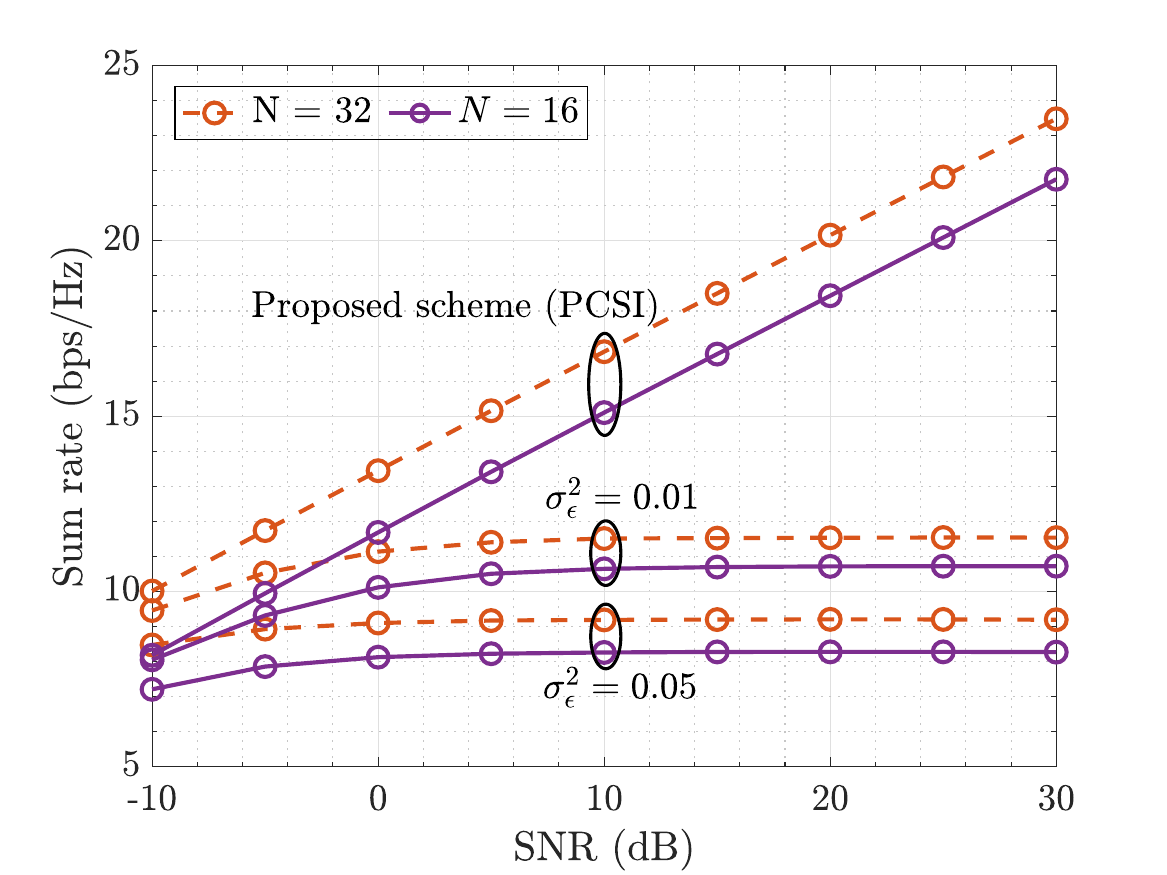}
    \caption{Sum rate vs. SNR}
    \label{SR_SNR}  
\end{figure}

Figure \ref{SR_SNR} shows the impact of SNR on the sum rate  performance under ICSI for different values of channel estimation error variance $\sigma_\epsilon^2$ and the number of RIS elements $N$. We observe that the sum rate  is smaller when the variance of ICSI is higher which is expected. This degradation can be overcome by increasing the number of RIS elements, as can be seen from the figure.  
Further, it can be observed that the sum rate saturates with an increase in SNR since SINR becomes limited by the channel estimation error at higher SNR values. This is more evident from the PCSI result when $\sigma^2_{\epsilon} = 0$ that increases linearly with SNR. It can also be seen that the sum rate  performance of PCSI and ICSI schemes are very close at low SNR due to the noise-limited behavior of the system. 
Furthermore, with increasing number of RIS elements $N$, the impact of noise can be regulated at lower SNR because of the array gain. However, from \eqref{ICSI_Power}, it can be observed that the overall channel estimation error variance also increases with $N$. From this, it can be concluded that the sum rate performance is almost independent of SNR for a given $\sigma^2_{\epsilon}$ as can be verified from the figure. 
Finally, based on the observations made from the figure, we can remark that the sum rate  performance is limited by noise in the low SNR region and by CSI error in the high SNR region. Regardless, these limits can be pushed by increasing the number of RIS elements.   
\begin{figure}[h!]
    \centering 
    \includegraphics[width=.6\columnwidth]{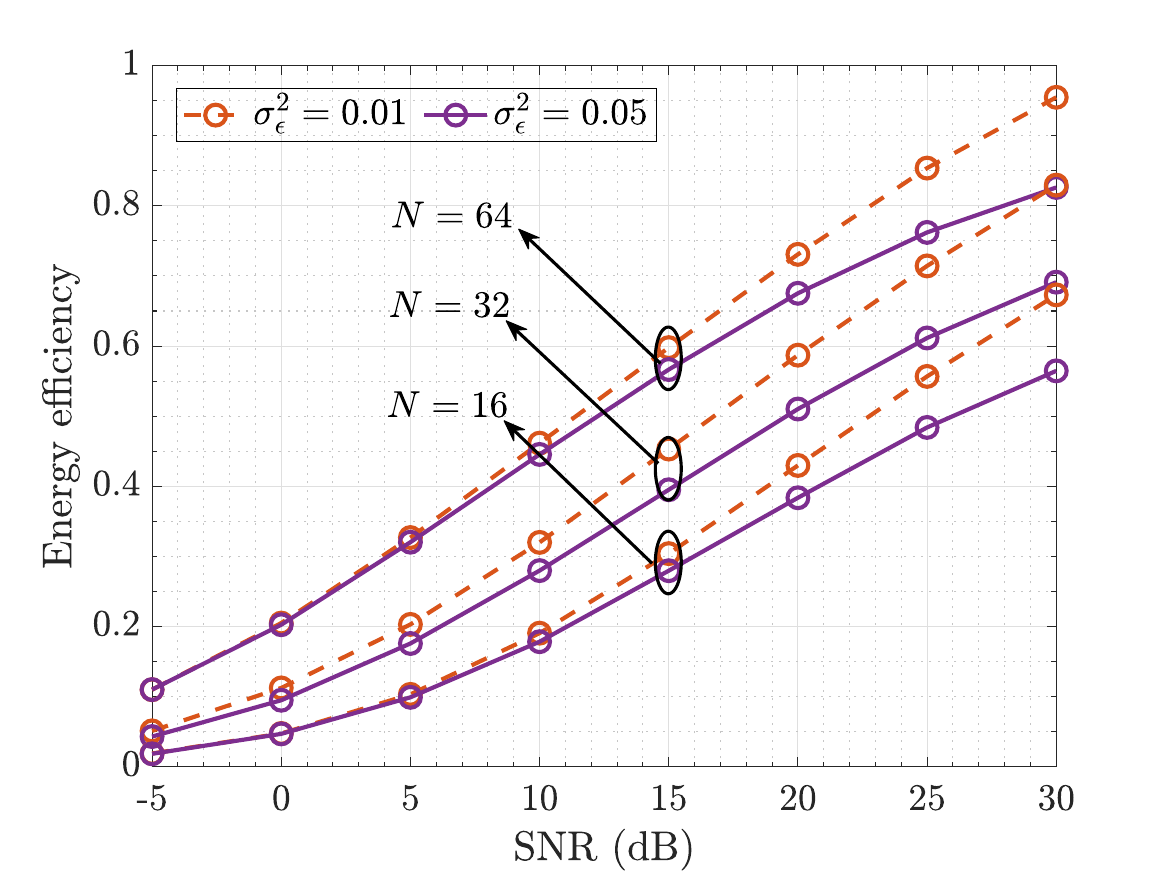}
    \caption{Energy efficiency vs. SNR}
    \label{EE_SNR}
\end{figure}

Figure \ref{EE_SNR} shows the energy efficiency performance as a function of SNR under ICSI for various values of $N$ and $\sigma_\epsilon^2$. We observe that energy efficiency improves as $N$ increases due to the improvement in SNR from the increased array gain as $N$ increases. This improvement in SNR will result in reduced transmission power requirements to maintain minimum rate constraints for each user while achieving sum rate. Further, it can be seen that the energy efficiency improves with a drop in $\sigma_\epsilon^2$. This is due to the improved sum rate performance with the decrease in $\sigma_\epsilon^2$, as can be seen in the figure.

\begin{figure}[h!]
    \centering
    \includegraphics[width=.6\columnwidth]{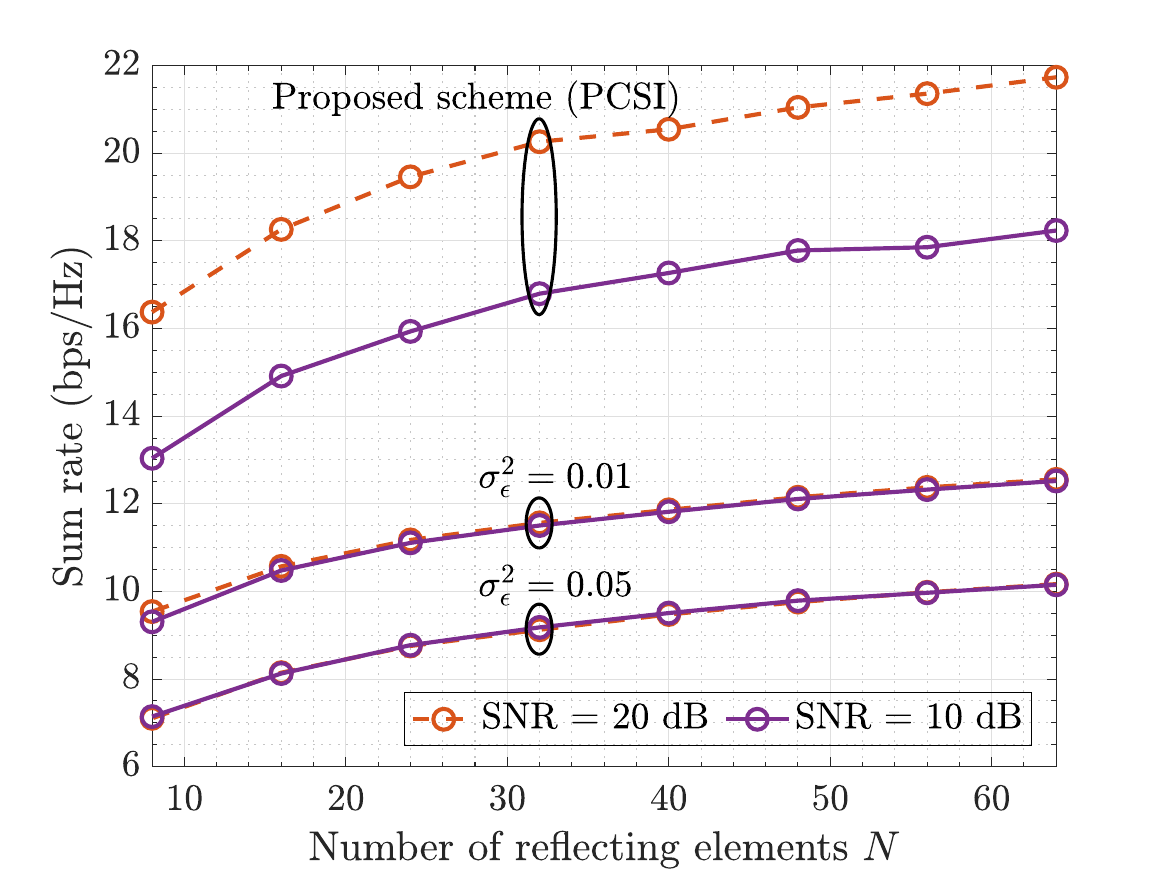}
    \caption{Sum rate vs. number of RIS elements $N$}
    \label{SR_N}
\end{figure}
Figure \ref{SR_N} shows that the sum rate  performance improves with respect to the number of RIS elements $N$ for various values of SNR and $\sigma_\epsilon^2$. 
The figure also shows that the sum rate  performance is more sensitive to SNR as $\sigma_\epsilon^2$ decreases toward the perfect CSI scheme with $\sigma_\epsilon^2 = 0$. This is because the noise power dominates the sum rate performance in the absence of channel estimation error variance. To elaborate further, there is a noticeable gap between the curves at different values of SNR when $\sigma_\epsilon^2 = 0$. On the other hand, the gap is negligible when $\sigma_\epsilon^2$ is higher. This is also evident from the comparison between PCSI and ICSI curves at different SNR values. Besides, the rate of increase in sum rate with respect to $N$ is also sensitive to $\sigma^2_{\epsilon}$ in a similar manner as discussed above. 
\begin{figure}[h!]
    \centering 
    \includegraphics[width=.6\columnwidth]{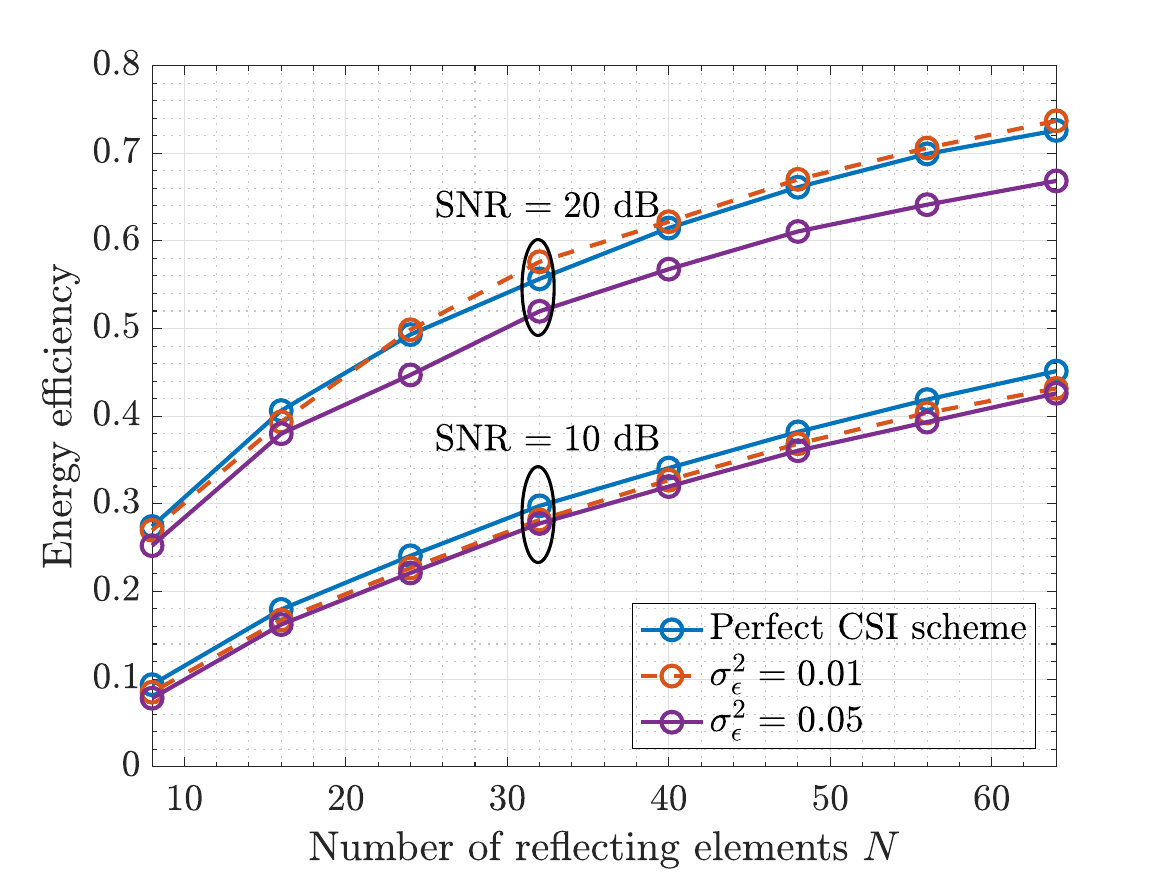}
    \caption{Energy efficiency vs. number of RIS elements $N$}
    \label{EE_N}
\end{figure}

Figure \ref{EE_N} shows that the energy efficiency performance improves consistently with increasing $N$. This is mainly due to larger RIS providing larger array gains, which aids in improving the sum rate  (while ensuring the minimum rate constraint) using a small transmission power. In essence, a higher sum rate  can be obtained at a small transmission power (or, equivalently, low SNR) with the increase in $N$. This subsequently leads to a significant gain in energy efficiency, as can be seen from the figure. Unlike the sum rate, energy efficiency is sensitive to SNR and does not vary much with respective to $\sigma^2_{\epsilon}$. 
In addition, it can also be observed that the energy efficiency is affected significantly with respect to  $\sigma^2_{\epsilon}$ at higher SNR.  

\begin{figure}[h!]
    \centering 
    \includegraphics[width=.6\columnwidth]{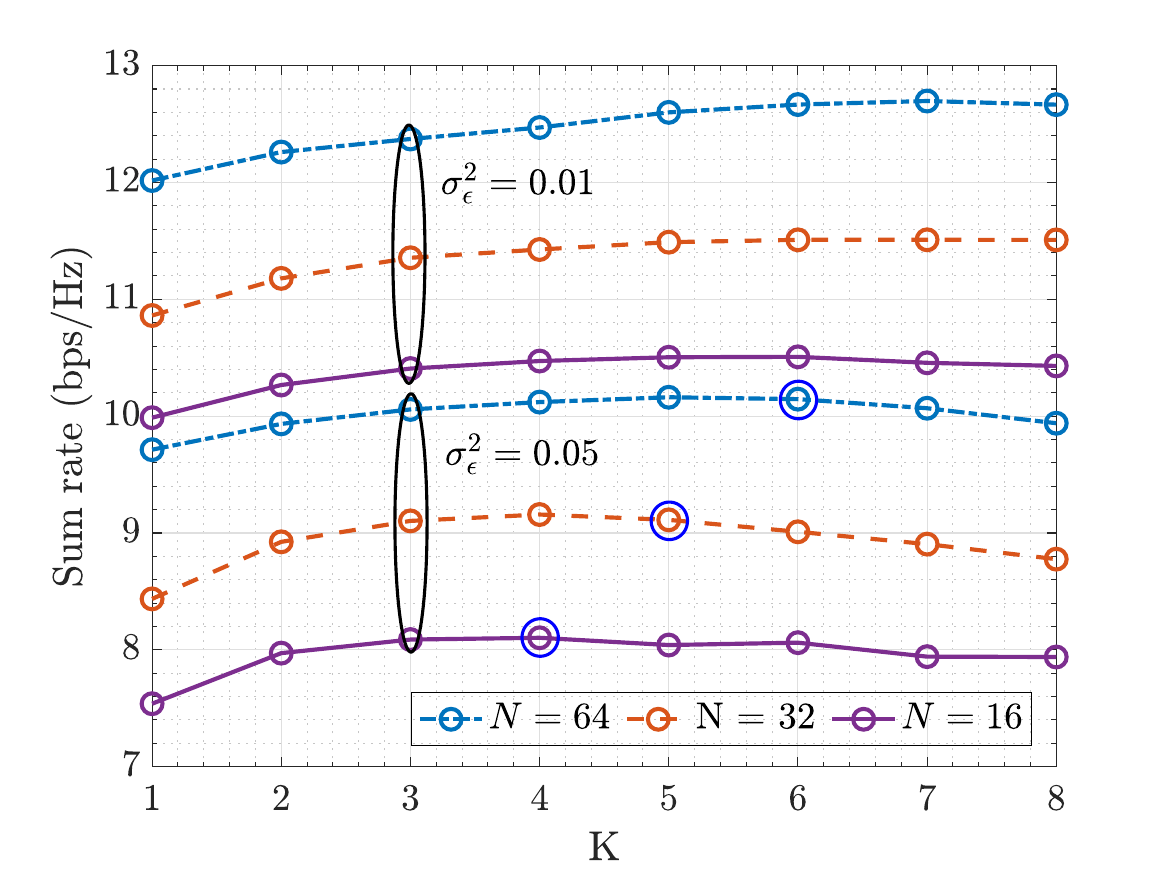}
    \caption{Sum rate vs. number of users $K$}
    \label{SR_K}
\end{figure}

Figure \ref{SR_K} shows the impact of the number of users $K$ on the sum rate  performance for different values of $N$ and $\sigma^2_{\epsilon}$. It can be seen that the sum rate increases with respect to $K$ and eventually drops. This improvement is expected because of the increased channel diversity with increasing $K$. However, the degradation at higher $K$ is attributed to the fact associated with the minimum rate constraint and increasing inter-user interference in the presence of NOMA. Moreover, this degradation starts at higher $K$ when $\sigma^2_{\epsilon}$ is small and when $N$ is large. These starting points of degradation for $\sigma^2_{\epsilon} = 0.05$ have been roughly circled in the figure.     
\begin{figure}[h!]
    \centering 
    \includegraphics[width=.6\columnwidth]{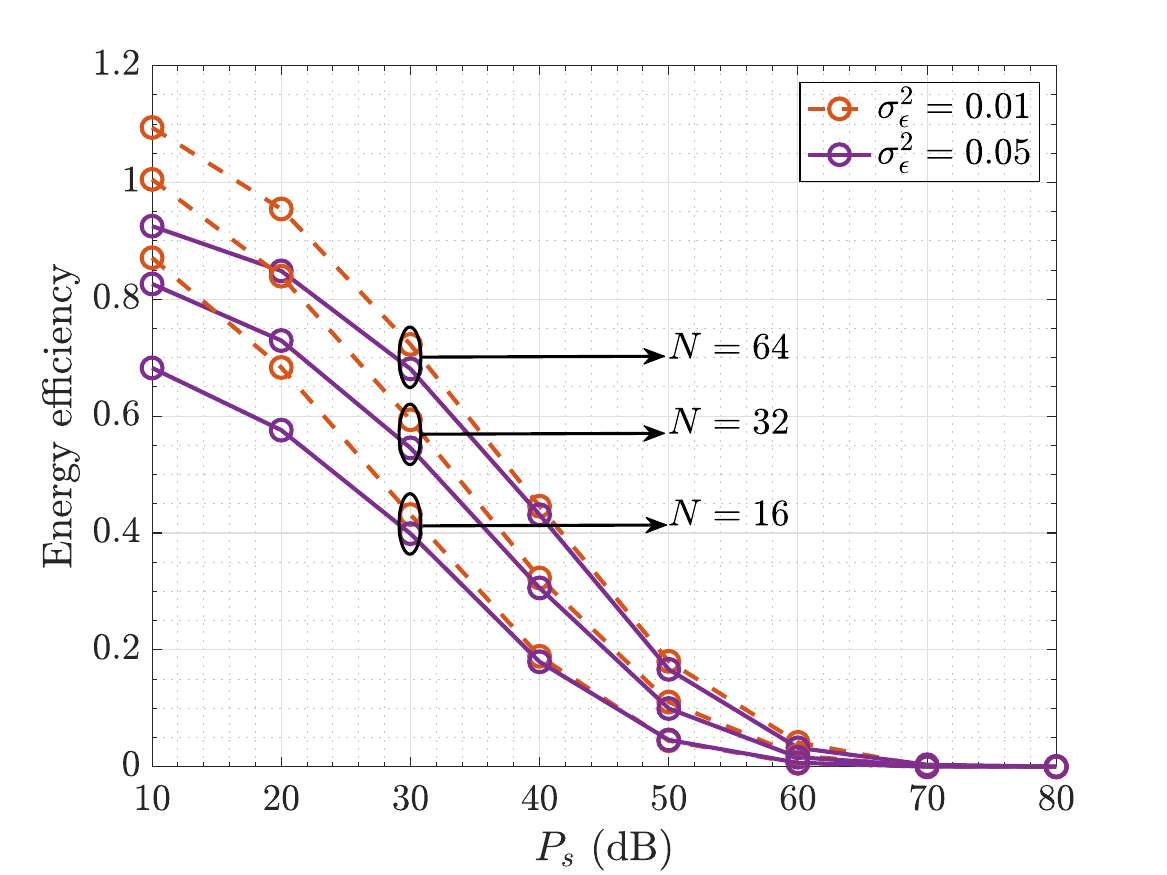}
    \caption{Energy efficiency vs. total transmit power $(P_s),~\rm{dB}$}
    \label{EE_Ps}
\end{figure}

The impact of transmission power $P_s$  on the energy efficiency performance is shown in Figure \ref{EE_Ps}. The figure shows that the energy efficiency performance degrades slowly with increasing $P_s$. This is because the numerator and the denominator of energy efficiency increase logarithmically and linearly with respect to $P_s$. Another interesting observation is for a fixed value of $P_s$, the energy efficiency improves as $N$ increases. This may be attributed to the fact that the array gain associated with RIS improves as $N$ increases, as discussed before, which results in low power consumption to meet the minimum rate requirements of all the users. 
\begin{figure}[h!]
    \centering 
    \includegraphics[width=.6\columnwidth]{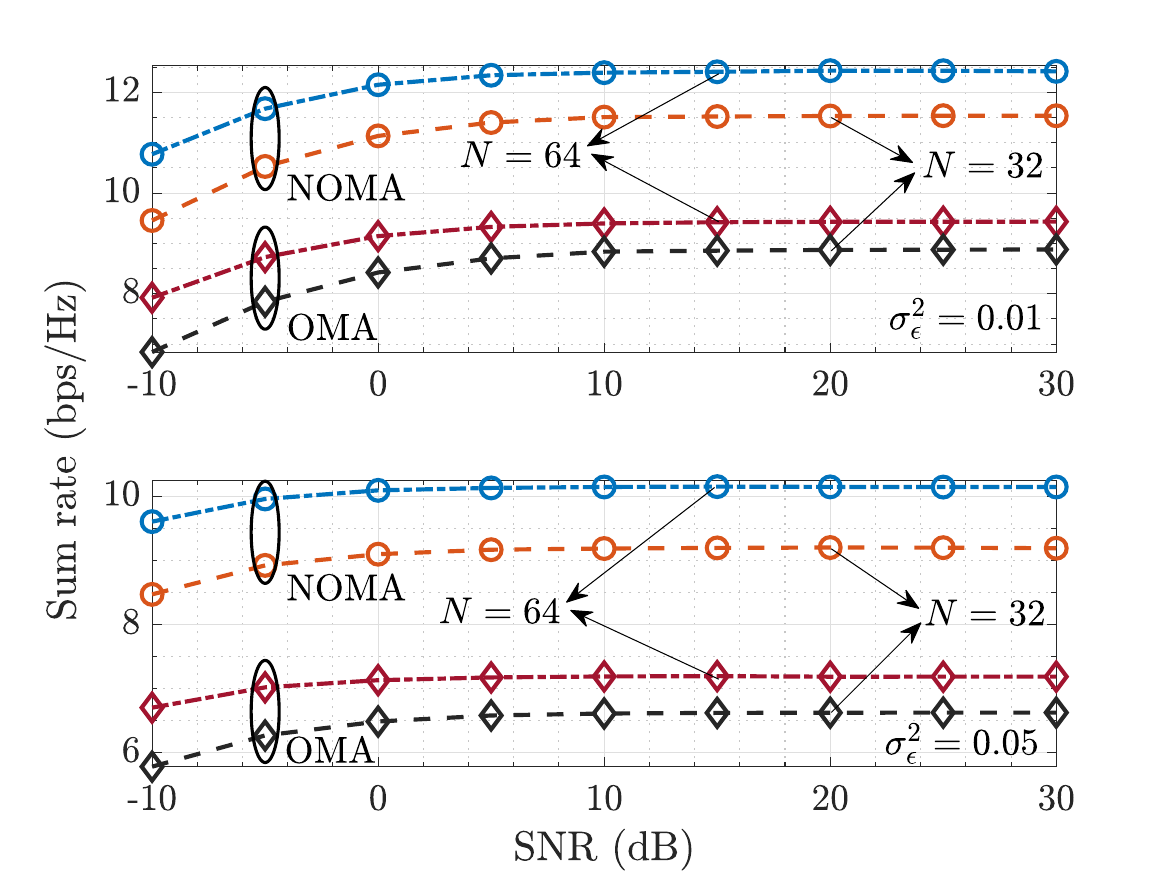}
    \caption{NOMA vs. OMA comparison}
    \label{SR_NOMA_OMA}
\end{figure}

In Figure \ref{SR_NOMA_OMA}, we provide the comparison between the achievable sum rate under NOMA and orthogonal multiple access (OMA) schemes. It can be observed from the figure that NOMA outperforms OMA for a wide range of $N$ and $\sigma^2_{\epsilon}$ values. This highlights the sum rate gain of a communication system enabled by spatial multiplexing. Further, we also observe that $\sigma^2_{\epsilon}$ impacts NOMA performance more in comparison to OMA.

\begin{figure}[h!]
\centering
\begin{minipage}{.5\textwidth}
  \centering
  \includegraphics[width=\columnwidth]{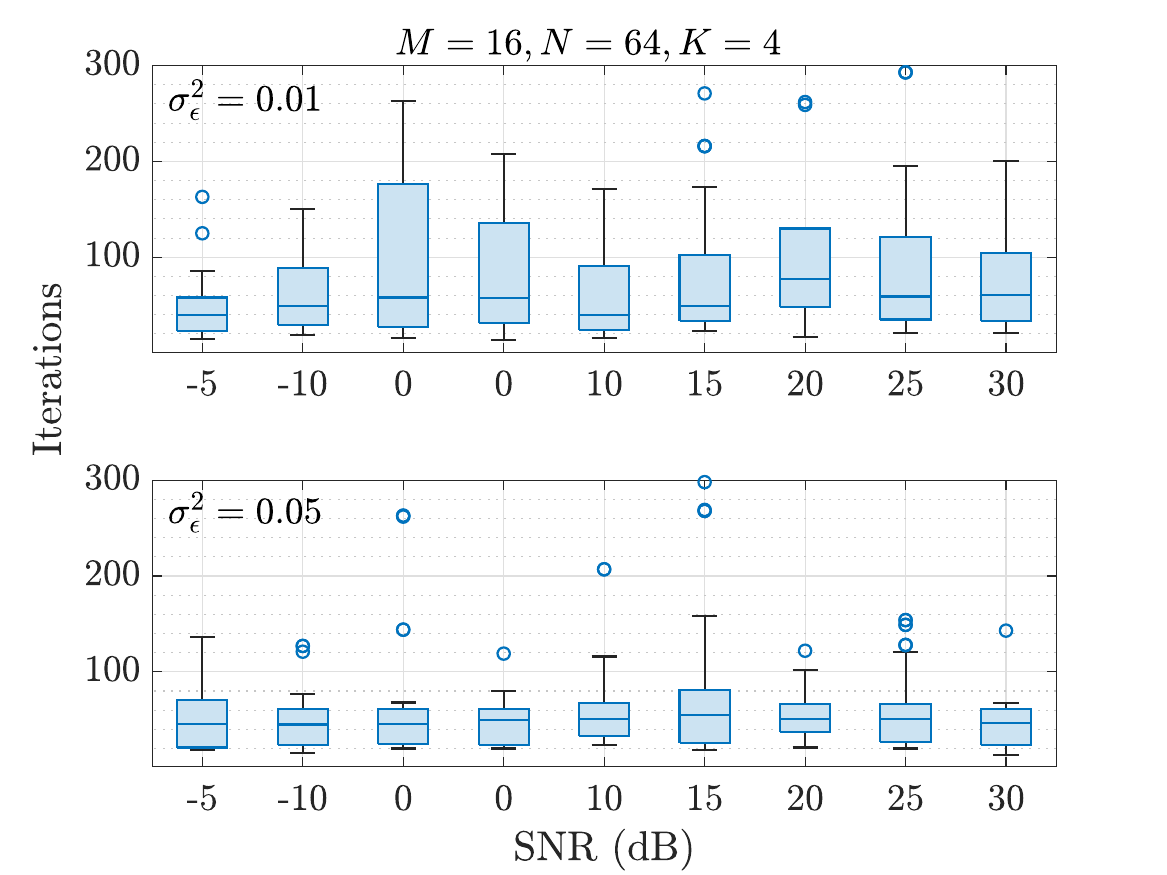}
\end{minipage}%
\begin{minipage}{.5\textwidth}
  \centering
  \includegraphics[width=\columnwidth]{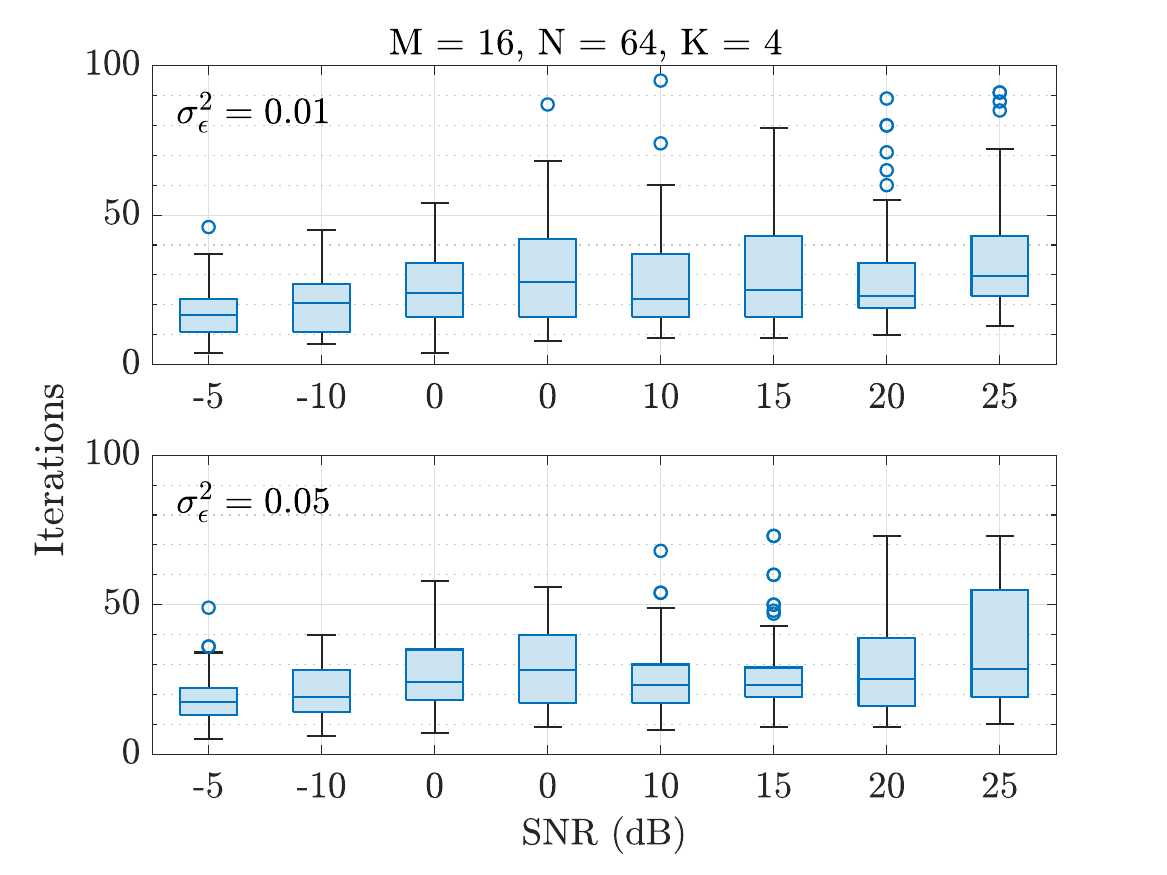}
\end{minipage}
\caption{Convergence of the proposed Quadratic transform-based  Algorithm 1 (left) and Algorithm 2 (right).}
\label{Convergance}
\end{figure}

Figure \ref{Convergance} shows the convergence performance of the proposed algorithms for sum rate and energy efficiency maximization. In particular, the figure visually shows the spread of the number of iterations both algorithms take to converge for $50$ channel realizations. 
The figure shows that the number of iterations required for Algorithm 1 to converge is larger in comparison to Algorithm 2.
Next, we can observe that the median value of iterations of both algorithms remains approximately the same irrespective of the SNR value. Besides, it can be observed that the spread of the number of iterations is larger when $\sigma^2_{\epsilon}$ is smaller, particularly for the sum rate maximization algorithm.
\section{Conclusion}
In this paper, we investigate RIS integration in communication systems with low-complexity transceiver architecture. Towards this, we propose a fully analog architecture at the BS, transmitting information to multiple users. Such a system has the limitation of single-user transmission due to the single RF chain at the transmitter. Thus, to overcome this drawback, we employ NOMA to enable multi-user transmission, thereby improving the spatial multiplexing gains of such a low-complexity system. In particular, we study the sum rate and energy efficiency performances of the proposed system integrated with RIS. Thus, we formulate independent sum rate  and energy efficiency maximization problems to obtain the joint transmit beamformer, RIS phase shifts, and power allocation solutions. However, these problems are non-convex due to 1) fractional SINR term, 2) non-convex minimum rate and unit modulus RIS phase shift constraints, and 3) coupled variables. To tackle these problems, we propose a framework with three main steps: 1) apply the quadratic transform to decouple the numerator and denominator of the SINR term, 2) employ SCA and SDR to simplify the non-convex minimum rate and unit modulus phase shifts constraints respectively, and 3) AO-based iterative algorithm to decouple the transmit beamformer, RIS phase shifts, and power allocation variables. Moreover, the transmit beamformer and RIS phase shift matrix subproblems are equivalent under sum rate  and energy efficiency problems. Further, we also establish the equivalence between the quadratic transform and weighted MSE minimization under the transmit beamformer and RIS phase shift matrix subproblems. Through extensive numerical analysis, we show that the performance of the proposed RIS-aided low-complexity system outperforms the capacity-achieving fully digital system with an SVD-based precoder employing a filling algorithm for power allocation. 

\section{Future scope}
The results presented in this paper aim to be a starting point for our investigation on RIS as a solution to reduce transceiver complexity. This new RIS use case can be a remarkable solution to creating sustainable and cost-efficient systems. Our goal is to solidify further the open problem we pose as ``\textit{Can RIS reduce transceiver complexity?}" through our future works. 

\begin{appendices}
\section{Weighted MSE and Quadratic Transform Equivalence}\label{Appendix_QT_WMSE_equiv}
We recall the objective function of the sum rate maximization problem \eqref{SR} as 
$$\sum_{k = 1}^K R_k = \sum_{k = 1}^K \log_2(1+\Gamma_k).$$
We first write the WMSE reformulation of the sum rate maximization problem as follows. The mean square error is given as 
\begin{align}
e_k & = \mathbb{E} \left[ |u_ky_k - s_k|^2\right] , \nonumber \\
&=  \mathbb{E} [ |u_k \Big(\widehat{\mathbf{h}}_{k}^H \mathbf{\Phi} \widehat{\mathbf{H}}\mathbf{f}\sqrt{p_k} s_k + \sum_{i=k+1}^{K} \widehat{\mathbf{h}}_{k}^H \mathbf{\Phi} \widehat{\mathbf{H}}\mathbf{f} \sqrt{p_i} s_i \nonumber \\ &~~+ \mathbf{v}_{k}^H \sum_{i=1}^{K}\sqrt{p_i}  s_i + n_k\Big) - s_k |^2], \nonumber \\
&= 1 + \sigma^2|u_k|^2 - 2\Re\{u_k\widehat{\mathbf{h}}_{k}^H \mathbf{\Phi} \widehat{\mathbf{H}}\mathbf{f}\sqrt{p_k}\} \nonumber \\ &~~+ \sum_{i=k}^K|u_k\widehat{\mathbf{h}}_{k}^H \mathbf{\Phi} \widehat{\mathbf{H}}\mathbf{f}\sqrt{p_i}|^2 + |u_k|^2\alpha\sum_{i = 1}^{K} p_i.\label{MSE_opt}
\end{align} 
where $u_k$ represents the MSE receiver. 
Using \eqref{MSE_opt}, the equivalent WMSE reformulation of the sum rate maximization problem given in \eqref{SR} can be written as 
\begin{subequations}
    \begin{align} 
\min_{w_k,u_k,\mathbf{f},\boldsymbol{\psi},\mathbf{p}} \quad &\sum_{k=1}^{K} w_k e_k - \log_2 w_k,\label{WMSE_Obj} \\
\textrm{s.t.} \quad &\eqref{SR_Phi_Con}-\eqref{SR_MinRate_Con},%
\end{align}\label{WMSE_Opti}%
\end{subequations}
where $w_k$ represents the weights.  
Here the WMSE minimization equivalent is a function of $w_k$ and $u_k$  whose optimal values \cite{WMSE_SR_Equiv} can be obtained as $$w_k^{\rm{opt}} = (e_k^{\rm{mmse}})^{-1}$$ and $$u_k^{\rm{opt}} = \Big[ \sigma^2 + \sum_{i=k}^K |\widehat{\mathbf{h}}_{k}^H \mathbf{\Phi} \widehat{\mathbf{H}}\mathbf{f}\sqrt{p_i}|^2 + \alpha_k\sum_{i = 1}^{K}p_i \Big]^{-1} \sqrt{p_k}\mathbf{f}^H\widehat{\mathbf{H}}^H\mathbf{\Phi}^H\widehat{\mathbf{h}}_{k},$$ where $e_k^{\rm{mmse}}$ is minimum MSE (MMSE) that is obtained by substituting $u_k^{\rm{opt}}$ in $e_k$ as 
    \begin{align}
        e_k^{\rm{mmse}} &= 1 - \frac{|\widehat{\mathbf{h}}_{k}^H \mathbf{\Phi} \widehat{\mathbf{H}}\mathbf{f}\sqrt{p_k}|^2}{\sigma^2 + \sum_{i=k}^K |\widehat{\mathbf{h}}_{k}^H \mathbf{\Phi} \widehat{\mathbf{H}}\mathbf{f}\sqrt{p_i}|^2 + \alpha_k\sum_{i = 1}^{K}p_i}, \nonumber\\
        &= \frac{1}{1+\Gamma_k}
    \end{align}
Now, by substituting $w_k^{\rm{opt}}$ and $e_k^{\rm{mmse}}$ in \eqref{WMSE_Obj}, the objective becomes 
    \begin{align}
    \sum_{k=1}^K w_k^{\rm{opt}} e_k^{\rm{mmse}} - \log_2 w_k^{\rm{opt}} &= \sum_{k=1}^K 1 - \log_2 \{e_k^{\rm{mmse}}\}^{-1} \nonumber \\  
    &= \sum_{k=1}^K 1 - \log_2 \{1+\Gamma_k\} \label{WMSE_to_SR}. 
\end{align}
This ensures that the minimization of WMSE with auxiliary variables $w_k$ and $u_k$ is equivalent to maximizing the sum rate.  

We now show the equivalence of WMSE reformulation and the quadratic transform-based reformulation. Towards this, we change the variable $u_k$ to $\Tilde{u}_k$ and MSE $e_k$ to $e_k^{\rm{QT}}$, and rewrite $e_k^{\rm{QT}}$  as 
\begin{align}
e_k^{\rm{QT}} &= - \Bigg( -1 + 2\Re\{\Tilde{u}_k^*\widehat{\mathbf{h}}_{k}^H \mathbf{\Phi} \widehat{\mathbf{H}}\mathbf{f}\sqrt{p_k}\} \nonumber \\ &- |\Tilde{u}_k|^2 \{\sigma^2 + |\widehat{\mathbf{h}}_{k}^H \mathbf{\Phi} \widehat{\mathbf{H}}\mathbf{f}|^2 \sum_{i=k}^K p_i + \alpha \sum_{i=1}^K p_i\} \Bigg)\label{ek_IQT_eq}
\end{align}
The above expression is quadratic in  $\Tilde{u}_k$ and concave in optimization variables $\boldsymbol{\psi}$/$\mathbf{f}$/$\mathbf{p}$, and can be thought of as the quadratic transform with $\Tilde{u}_k$ being an auxiliary variable. We intuitively revert \eqref{ek_IQT_eq} back to its corresponding fractional representation  as 
\begin{align}
e_k^{\rm{IQT}} &= - \Bigg[ -1 + \frac{|\widehat{\mathbf{h}}_{k}^H \mathbf{\Phi} \widehat{\mathbf{H}}\mathbf{f}\sqrt{p_k}|^2}{\sigma^2 + |\widehat{\mathbf{h}}_{k}^H \mathbf{\Phi} \widehat{\mathbf{H}}\mathbf{f}|^2 \sum_{i=k}^K p_i + \alpha \sum_{i=1}^K p_i}\Bigg] \nonumber\\ \nonumber
&=  \Bigg[ \frac{\sigma^2 + |\widehat{\mathbf{h}}_{k}^H \mathbf{\Phi} \widehat{\mathbf{H}}\mathbf{f}|^2 \sum_{i=k+1}^K p_i + \alpha \sum_{i=1}^K p_i}{\sigma^2 + |\widehat{\mathbf{h}}_{k}^H \mathbf{\Phi} \widehat{\mathbf{H}}\mathbf{f}|^2 \sum_{i=k}^K p_i + \alpha \sum_{i=1}^K p_i}  \Bigg] \\ \nonumber
&= \frac{1}{1+\frac{|\widehat{\mathbf{h}}_{k}^H \mathbf{\Phi} \widehat{\mathbf{H}}\mathbf{f}\sqrt{p_k}|^2}{\sigma^2 + |\widehat{\mathbf{h}}_{k}^H \mathbf{\Phi} \widehat{\mathbf{H}}\mathbf{f}|^2 \sum_{i=k+1}^K p_i + \alpha \sum_{i=1}^K p_i}}\nonumber \\ 
&= \frac{1}{1+\Gamma_k} \label{e_IQT},
\end{align}
Further, by substituting \eqref{e_IQT} in \eqref{WMSE_Obj}, we obtain 
\begin{subequations}
\begin{align} 
\min_{w_k,u_k,\mathbf{f},\boldsymbol{\psi},\mathbf{p}} \quad & \sum_{k=1}^K w_k e_k^{\rm{IQT}} - \log_2 \{e_k^{\rm{IQT}}\}^{-1} \\
\textrm{s.t.} \quad &\eqref{SR_Phi_Con}-\eqref{SR_MinRate_Con}.%
\end{align}\label{WMSE_QT_1}%
\end{subequations}
The above objective reduces to $\sum_{k=1}^K 1 - \log_2(1+\Gamma_k)$ for $w_k^{\rm{opt}} = (e_k^{\rm{IQT}})^{-1}$, which makes \eqref{WMSE_QT_1} equivalent to the sum rate maximization problem as follows 
\begin{subequations}
\begin{align} 
\max_{\mathbf{f},\boldsymbol{\psi},\mathbf{p}} \quad & \sum_{k=1}^K  \log_2(1+\Gamma_k) \\
\textrm{s.t.} \quad &\eqref{SR_Phi_Con}-\eqref{SR_MinRate_Con},%
\end{align}%
\end{subequations}
From this, we can deduce that WMSE reformulation is equivalent to the quadratic transform for sum rate maximization problem.  
In summary, we first present the equivalence of sum rate maximization to the WMSE reformulation with auxiliary variables $w_k$ and $u_k$. Next, we revert this MSE equation back to the equivalent fractional form using the quadratic transform with the auxiliary variable $\Tilde{u}_k$.  
The reverted function that is in the fractional form is equivalent to the original objective function of the sum rate maximization problem.
Using this argument, we deduce the equivalence between the WMSE reformulation and the quadratic transform based on the fact that the sum rate expression is obtained by reverting the MSE expression into the equivalent fractional form based on the quadratic transform.
\end{appendices}
\endgroup
\bibliographystyle{IEEEtran}
\bibliography{IEEEabrv,IEEE-Transactions-LaTeX2e-templates-and-instructions/Ref_New}

\end{document}